\documentclass{eptcs}

\usepackage{hyperref}
\usepackage{amsmath,amsthm,amssymb}
\usepackage{stmaryrd,cmll}
\usepackage{proof,linear}

\usepackage{xspace,enumerate,color,epsfig}
\usepackage{graphicx}
\graphicspath{{.}{./figures/}}
\usepackage{marvosym}
\usepackage{bm}

\usepackage{tikzfig}
\usepackage{docmute}
\usepackage{keycommand}

\newcommand{\emptydiag}{\,\tikz{\node[style=empty diagram] (x) {};}\,}

\newkeycommand{\boxmap}[style=small box][1]{\,\tikz{\node[style=\commandkey{style}] (x) {$#1$};\draw(0,-1.25)--(x)--(0,1.25);}\,}
\newkeycommand{\dboxmap}[style=dbox][1]{\,\tikz{\node[style=\commandkey{style}] (x) {$#1$};\draw[boldedge] (0,-1.25)--(x)--(0,1.25);}\,}

\newkeycommand{\wirelabel}[style=white label]{\,\tikz{\node[style=\commandkey{style}]{};}\,\xspace}

\newkeycommand{\pointmap}[style=point][1]{\,\tikz{\node[style=\commandkey{style}] (x) at (0,-0.3) {$#1$};\node [style=none] (3) at (0, 0.9) {};\node [style=none] (3) at (0, -0.9) {};\draw (x)--(0,0.7);}\,}

\newkeycommand{\point}[style=point,estyle=][1]{\,\tikz{\node[style=\commandkey{style}] (x) at (0,-0.05) {$#1$};\draw [\commandkey{estyle}] (x)--(0,1.05);}\,}

\newkeycommand{\copointmap}[style=copoint][1]{\,\tikz{\node[style=\commandkey{style}] (x) at (0,0.3) {$#1$};\draw (x)--(0,-0.7);}\,}

\newkeycommand{\dpoint}[style=point][1]{\,\tikz{\node[style=\commandkey{style},doubled] (x) at (0,-0.3) {$#1$};\draw[boldedge] (x)--(0,0.7);}\,}

\newkeycommand{\kpoint}[style=kpoint,estyle=][1]{\,\tikz{\node[style=\commandkey{style}] (x) at (0,-0.05) {$#1$};\draw [\commandkey{estyle}] (x)--(0,1.05);}\,}

\newkeycommand{\kpointgray}[style=gray kpoint,estyle=][1]{\,\tikz{\node[style=\commandkey{style}] (x) at (0,-0.05) {$#1$};\draw [\commandkey{estyle}] (x)--(0,1.05);}\,}

\newkeycommand{\typedkpoint}[style=kpoint,edgestyle=][2]{%
\,\begin{tikzpicture}
  \begin{pgfonlayer}{nodelayer}
    \node [style=none] (0) at (0, 1) {};
    \node [style=\commandkey{style}] (1) at (0, -0.75) {$#2$};
    \node [style=right label] (2) at (0.25, 0.5) {$#1$};
  \end{pgfonlayer}
  \begin{pgfonlayer}{edgelayer}
    \draw [\commandkey{edgestyle}] (1) to (0.center);
  \end{pgfonlayer}
\end{tikzpicture}\,}

\newkeycommand{\typedkpointdag}[style=kpoint adjoint,edgestyle=][2]{%
\,\begin{tikzpicture}
  \begin{pgfonlayer}{nodelayer}
    \node [style=none] (0) at (0, -1) {};
    \node [style=\commandkey{style}] (1) at (0, 0.75) {$#2$};
    \node [style=right label] (2) at (0.25, -0.5) {$#1$};
  \end{pgfonlayer}
  \begin{pgfonlayer}{edgelayer}
    \draw [\commandkey{edgestyle}] (0.center) to (1);
  \end{pgfonlayer}
\end{tikzpicture}\,}

\newkeycommand{\bistate}[style=kpoint][1]{\,\begin{tikzpicture}
  \begin{pgfonlayer}{nodelayer}
    \node [style=\commandkey{style}, minimum width=1 cm, inner sep=2pt] (0) at (0, -0.25) {$#1$};
    \node [style=none] (1) at (-0.75, 0) {};
    \node [style=none] (2) at (0.75, 0) {};
    \node [style=none] (3) at (-0.75, 0.88) {};
    \node [style=none] (4) at (0.75, 0.88) {};
  \end{pgfonlayer}
  \begin{pgfonlayer}{edgelayer}
    \draw (1.center) to (3.center);
    \draw (2.center) to (4.center);
  \end{pgfonlayer}
\end{tikzpicture}\,}

\newkeycommand{\bistatebig}[style=kpoint][1]{\,\begin{tikzpicture}
  \begin{pgfonlayer}{nodelayer}
    \node [style=\commandkey{style}, minimum width=, minimum width=1.26 cm, inner sep=2pt] (0) at (0, -0.25) {$#1$};
    \node [style=none] (1) at (-1, 0) {};
    \node [style=none] (2) at (1, 0) {};
    \node [style=none] (3) at (-1, 0.88) {};
    \node [style=none] (4) at (1, 0.88) {};
  \end{pgfonlayer}
  \begin{pgfonlayer}{edgelayer}
    \draw (1.center) to (3.center);
    \draw (2.center) to (4.center);
  \end{pgfonlayer}
\end{tikzpicture}\,}

\newkeycommand{\dbistate}[style=dkpoint][1]{\,\begin{tikzpicture}
  \begin{pgfonlayer}{nodelayer}
    \node [style=\commandkey{style}, minimum width=1 cm, inner sep=2pt] (0) at (0, -0.25) {$#1$};
    \node [style=none] (1) at (-0.75, 0) {};
    \node [style=none] (2) at (0.75, 0) {};
    \node [style=none] (3) at (-0.75, 1.0) {};
    \node [style=none] (4) at (0.75, 1.0) {};
  \end{pgfonlayer}
  \begin{pgfonlayer}{edgelayer}
    \draw [boldedge] (1.center) to (3.center);
    \draw [boldedge] (2.center) to (4.center);
  \end{pgfonlayer}
\end{tikzpicture}\,}

\newkeycommand{\bistatebraket}[style1=kpoint,style2=kpointadj][2]{\,\begin{tikzpicture}
  \begin{pgfonlayer}{nodelayer}
    \node [style=\commandkey{style1}, minimum width=1 cm, inner sep=2pt] (0) at (0, -0.75) {$#1$};
    \node [style=none] (1) at (-0.75, -0.5) {};
    \node [style=none] (2) at (0.75, -0.5) {};
    \node [style=none] (3) at (-0.75, 0.5) {};
    \node [style=none] (4) at (0.75, 0.5) {};
    \node [style=\commandkey{style2}, minimum width=1 cm, inner sep=2pt] (5) at (0, 0.75) {$#2$};
  \end{pgfonlayer}
  \begin{pgfonlayer}{edgelayer}
    \draw (1.center) to (3.center);
    \draw (2.center) to (4.center);
  \end{pgfonlayer}
\end{tikzpicture}\,}

\newcommand{\sboxdiscard}[1]{\,\begin{tikzpicture}
  \begin{pgfonlayer}{nodelayer}
    \node [style=small box] (0) at (0, 0) {$#1$};
    \node [style=none] (1) at (0, -1.25) {};
    \node [style=upground] (2) at (0, 1.50) {};
  \end{pgfonlayer}
  \begin{pgfonlayer}{edgelayer}
    \draw (0) to (1.center);
    \draw (0) to (2);
  \end{pgfonlayer}
\end{tikzpicture}\,}

\newkeycommand{\dkpoint}[style=dkpoint][1]{\,\tikz{\node[style=\commandkey{style}] (x) at (0,-0.1) {$#1$};\draw[boldedge] (x)--(0,1);}\,}
\newkeycommand{\dkpointadj}[style=dkpointadj][1]{\,\tikz{\node[style=\commandkey{style}] (x) at (0,0.1) {$#1$};\draw[boldedge] (x)--(0,-1);}\,}
\newkeycommand{\dkpointtrans}[style=dkpointtrans][1]{\,\tikz{\node[style=\commandkey{style}] (x) at (0,0.1) {$#1$};\draw[boldedge] (x)--(0,-1);}\,}
\newkeycommand{\dkpointconj}[style=dkpointconj][1]{\,\tikz{\node[style=\commandkey{style}] (x) at (0,-0.1) {$#1$};\draw[boldedge] (x)--(0,1);}\,}

\newcommand{\trace}{\,\begin{tikzpicture}
  \begin{pgfonlayer}{nodelayer}
    \node [style=none] (0) at (0, -0.60) {};
    \node [style=upground] (1) at (0, 0.40) {};
  \end{pgfonlayer}
  \begin{pgfonlayer}{edgelayer}
    \draw [style=boldedge] (0.center) to (1);
  \end{pgfonlayer}
\end{tikzpicture}\,}

\newcommand\discard\trace

\newcommand{\sdiscard}{\,\begin{tikzpicture}
  \begin{pgfonlayer}{nodelayer}
    \node [style=none] (0) at (0, -0.60) {};
    \node [style=upground] (1) at (0, 0.40) {};
  \end{pgfonlayer}
  \begin{pgfonlayer}{edgelayer}
    \draw (0.center) to (1);
  \end{pgfonlayer}
\end{tikzpicture}\,}

\newkeycommand{\pointketbra}[style1=copoint,style2=point][2]{\,%
\begin{tikzpicture}
  \begin{pgfonlayer}{nodelayer}
    \node [style=none] (0) at (0, -1.75) {};
    \node [style=\commandkey{style1}] (1) at (0, -0.75) {$#1$};
    \node [style=\commandkey{style2}] (2) at (0, 0.75) {$#2$};
    \node [style=none] (3) at (0, 1.75) {};
  \end{pgfonlayer}
  \begin{pgfonlayer}{edgelayer}
    \draw (0.center) to (1);
    \draw (2) to (3.center);
  \end{pgfonlayer}
\end{tikzpicture}\,}

\newkeycommand{\twocopointketbra}[style1=copoint,style2=copoint,style3=point][3]{\,%
\begin{tikzpicture}
  \begin{pgfonlayer}{nodelayer}
    \node [style=none] (0) at (-0.75, -1.75) {};
    \node [style=none] (1) at (0.75, -1.75) {};
    \node [style=\commandkey{style1}] (2) at (-0.75, -0.75) {$#1$};
    \node [style=\commandkey{style2}] (3) at (0.75, -0.75) {$#2$};
    \node [style=\commandkey{style3}] (4) at (0, 0.75) {$#3$};
    \node [style=none] (5) at (0, 1.75) {};
  \end{pgfonlayer}
  \begin{pgfonlayer}{edgelayer}
    \draw (0.center) to (2);
    \draw (1.center) to (3);
    \draw (4) to (5.center);
  \end{pgfonlayer}
\end{tikzpicture}\,}

\newkeycommand{\twopointketbra}[style1=copoint,style2=point,style3=point][3]{\,%
\begin{tikzpicture}
  \begin{pgfonlayer}{nodelayer}
    \node [style=none] (0) at (0, -1.75) {};
    \node [style=none] (1) at (-0.75, 1.75) {};
    \node [style=\commandkey{style1}] (2) at (0, -0.75) {$#1$};
    \node [style=\commandkey{style2}] (3) at (-0.75, 0.75) {$#2$};
    \node [style=\commandkey{style3}] (4) at (0.75, 0.75) {$#3$};
    \node [style=none] (5) at (0.75, 1.75) {};
  \end{pgfonlayer}
  \begin{pgfonlayer}{edgelayer}
    \draw (0.center) to (2);
    \draw (1.center) to (3);
    \draw (4) to (5.center);
  \end{pgfonlayer}
\end{tikzpicture}\,}

\newkeycommand{\pointbraket}[style1=point,style2=copoint,estyle=][2]{\,%
\begin{tikzpicture}
  \begin{pgfonlayer}{nodelayer}
    \node [style=\commandkey{style1}] (0) at (0, -0.625) {$#1$};
    \node [style=\commandkey{style2}] (1) at (0, 0.625) {$#2$};
  \end{pgfonlayer}
  \begin{pgfonlayer}{edgelayer}
    \draw [\commandkey{estyle}] (0) to (1);
  \end{pgfonlayer}
\end{tikzpicture}\,}

\newkeycommand{\kpointketbra}[style1=kpointdag,style2=kpoint,estyle=][2]{\,%
\begin{tikzpicture}
  \begin{pgfonlayer}{nodelayer}
    \node [style=none] (0) at (0, -1.9) {};
    \node [style=\commandkey{style1}] (1) at (0, -0.95) {$#1$};
    \node [style=\commandkey{style2}] (2) at (0, 0.95) {$#2$};
    \node [style=none] (3) at (0, 1.9) {};
  \end{pgfonlayer}
  \begin{pgfonlayer}{edgelayer}
    \draw [\commandkey{estyle}] (0.center) to (1);
    \draw [\commandkey{estyle}] (2) to (3.center);
  \end{pgfonlayer}
\end{tikzpicture}\,}

\newkeycommand{\kpointmap}[style1=kpoint,style2=map][2]{\,%
\begin{tikzpicture}
  \begin{pgfonlayer}{nodelayer}
    \node [style=\commandkey{style1}] (0) at (0, -1.1) {$#1$};
    \node [style=\commandkey{style2}] (1) at (0, 0.5) {$#2$};
    \node [style=none] (2) at (0, 1.75) {};
  \end{pgfonlayer}
  \begin{pgfonlayer}{edgelayer}
    \draw (0) to (1);
    \draw (1) to (2.center);
  \end{pgfonlayer}
\end{tikzpicture}%
\,}

\newkeycommand{\dkpointmap}[style1=dkpoint,style2=dmap][2]{\,%
\begin{tikzpicture}
  \begin{pgfonlayer}{nodelayer}
    \node [style=\commandkey{style1}] (0) at (0, -1.2) {$#1$};
    \node [style=\commandkey{style2}] (1) at (0, 0.4) {$#2$};
    \node [style=none] (2) at (0, 1.7) {};
  \end{pgfonlayer}
  \begin{pgfonlayer}{edgelayer}
    \draw[style=boldedge] (0) to (1);
    \draw[style=boldedge] (1) to (2.center);
  \end{pgfonlayer}
\end{tikzpicture}%
\,}

\newkeycommand{\dkpointmapx}[style1=dkpoint,style2=dmap][2]{\,%
\begin{tikzpicture}
  \begin{pgfonlayer}{nodelayer}
    \node [style=\commandkey{style1}] (0) at (0, -1.4) {$#1$};
    \node [style=\commandkey{style2}] (1) at (0, 0.6) {$#2$};
    \node [style=none] (2) at (0, 1.9) {};
  \end{pgfonlayer}
  \begin{pgfonlayer}{edgelayer}
    \draw[style=boldedge] (0) to (1);
    \draw[style=boldedge] (1) to (2.center);
  \end{pgfonlayer}
\end{tikzpicture}%
\,}

\newkeycommand{\kpointsandwichmap}[style1=kpoint,style2=map,style3=kpointdag][3]{\,%
\begin{tikzpicture}
  \begin{pgfonlayer}{nodelayer}
    \node [style=\commandkey{style1}] (0) at (0, -1.5) {$#1$};
    \node [style=\commandkey{style2}] (1) at (0, 0) {$#2$};
    \node [style=\commandkey{style3}] (2) at (0, 1.5) {$#3$};
  \end{pgfonlayer}
  \begin{pgfonlayer}{edgelayer}
    \draw (0) to (1);
    \draw (1) to (2);
  \end{pgfonlayer}
\end{tikzpicture}%
}

\newkeycommand{\sandwichtwo}[style1=point,style2=map,style3=map,style4=copoint][4]{\,%
\begin{tikzpicture}
  \begin{pgfonlayer}{nodelayer}
    \node [style=\commandkey{style2}] (0) at (0, -1) {$#2$};
    \node [style=\commandkey{style3}] (1) at (0, 1) {$#3$};
    \node [style=\commandkey{style1}] (2) at (0, -2.6) {$#1$};
    \node [style=\commandkey{style4}] (3) at (0, 2.6) {$#4$};
  \end{pgfonlayer}
  \begin{pgfonlayer}{edgelayer}
    \draw (2) to (0);
    \draw (0) to (1);
    \draw (1) to (3);
  \end{pgfonlayer}
\end{tikzpicture}\,}

\newkeycommand{\longbraket}[style1=point,style2=copoint][2]{\,% 
\begin{tikzpicture}
  \begin{pgfonlayer}{nodelayer}
    \node [style=\commandkey{style1}] (0) at (0, -2.6) {$#1$};
    \node [style=\commandkey{style2}] (1) at (0, 2.6) {$#2$};
  \end{pgfonlayer}
  \begin{pgfonlayer}{edgelayer}
    \draw (0) to (1);
  \end{pgfonlayer}
\end{tikzpicture}\,}

\newkeycommand{\onb}[style=point]{\ensuremath{\{\,\tikz{\node[style=\commandkey{style}] (x) at (0,-0.3) {$j$};\draw (x)--(0,0.7);}\,\}}\xspace}

\newkeycommand{\phase}[style=white phase dot][1]{\,\begin{tikzpicture}
    \begin{pgfonlayer}{nodelayer}
        \node [style=none] (0) at (0, 1) {};
        \node [style=\commandkey{style}] (2) at (0, -0) {$#1$}; 
        \node [style=none] (3) at (0, -1) {};
    \end{pgfonlayer}
    \begin{pgfonlayer}{edgelayer}
        \draw (2) to (0.center);
        \draw (3.center) to (2);
    \end{pgfonlayer}
\end{tikzpicture}\,}

\newkeycommand{\dphase}[style=white phase ddot][1]{\,\begin{tikzpicture}
    \begin{pgfonlayer}{nodelayer}
        \node [style=none] (0) at (0, 1.25) {};
        \node [style=\commandkey{style}] (2) at (0, -0) {$#1$};
        \node [style=none] (3) at (0, -1.25) {};
    \end{pgfonlayer}
    \begin{pgfonlayer}{edgelayer}
        \draw [boldedge] (2) to (0.center);
        \draw [boldedge] (3.center) to (2);
    \end{pgfonlayer}
\end{tikzpicture}\,}

\newkeycommand{\dphasepoint}[style=white phase ddot][1]{\,\begin{tikzpicture}[yshift=-3mm]
  \begin{pgfonlayer}{nodelayer}
    \node [style=none] (0) at (0, 1) {};
    \node [style=\commandkey{style}] (1) at (0, 0) {$#1$};
  \end{pgfonlayer}
  \begin{pgfonlayer}{edgelayer}
    \draw [style=boldedge] (0.center) to (1);
  \end{pgfonlayer}
\end{tikzpicture}\,}

\newkeycommand{\dphasepointgray}[style=gray phase ddot][1]{\,\begin{tikzpicture}[yshift=-3mm]
  \begin{pgfonlayer}{nodelayer}
    \node [style=none] (0) at (0, 1) {};
    \node [style=\commandkey{style}] (1) at (0, 0) {$#1$};
  \end{pgfonlayer}
  \begin{pgfonlayer}{edgelayer}
    \draw [style=boldedge] (0.center) to (1);
  \end{pgfonlayer}
\end{tikzpicture}\,}

\newkeycommand{\dphasecopoint}[style=white phase ddot][1]{\,\begin{tikzpicture}[yshift=3mm,yscale=-1]
  \begin{pgfonlayer}{nodelayer}
    \node [style=none] (0) at (0, 1) {};
    \node [style=\commandkey{style}] (1) at (0, 0) {$#1$};
  \end{pgfonlayer}
  \begin{pgfonlayer}{edgelayer}
    \draw [style=boldedge] (0.center) to (1);
  \end{pgfonlayer}
\end{tikzpicture}\,}

\newkeycommand{\dphasepointsm}[style=white phase ddot][1]{\,\begin{tikzpicture}[yshift=-3mm]
  \begin{pgfonlayer}{nodelayer}
    \node [style=none] (0) at (0, 1) {};
    \node [style=\commandkey{style}] (1) at (0, 0) {\footnotesize\!$#1$\!};
  \end{pgfonlayer}
  \begin{pgfonlayer}{edgelayer}
    \draw [style=boldedge] (0.center) to (1);
  \end{pgfonlayer}
\end{tikzpicture}\,}

\newkeycommand{\phasepoint}[style=white phase dot][1]{\,\begin{tikzpicture}[yshift=-3mm]
  \begin{pgfonlayer}{nodelayer}
    \node [style=none] (0) at (0, 1) {};
    \node [style=\commandkey{style}] (1) at (0, 0) {$#1$};
  \end{pgfonlayer}
  \begin{pgfonlayer}{edgelayer}
    \draw (0.center) to (1);
  \end{pgfonlayer}
\end{tikzpicture}\,}

\newkeycommand{\phasecopoint}[style=white phase dot][1]{\,\begin{tikzpicture}[yshift=3mm]
  \begin{pgfonlayer}{nodelayer}
    \node [style=none] (0) at (0, -1) {};
    \node [style=\commandkey{style}] (1) at (0, 0) {$#1$};
  \end{pgfonlayer}
  \begin{pgfonlayer}{edgelayer}
    \draw (0.center) to (1);
  \end{pgfonlayer}
\end{tikzpicture}\,}

\newkeycommand{\kpointbraket}[style1=kpoint,style2=kpoint adjoint,estyle=][2]{\,%
\begin{tikzpicture}
  \begin{pgfonlayer}{nodelayer}
    \node [style=\commandkey{style1}] (0) at (0, -0.735) {$#1$};
    \node [style=\commandkey{style2}] (1) at (0, 0.735) {$#2$};
  \end{pgfonlayer}
  \begin{pgfonlayer}{edgelayer}
    \draw [\commandkey{estyle}] (0) to (1);
  \end{pgfonlayer}
\end{tikzpicture}\,}

\newkeycommand{\kpointbraketx}[style1=kpoint,style2=kpoint adjoint,estyle=][2]{\,%
\begin{tikzpicture}
  \begin{pgfonlayer}{nodelayer}
    \node [style=\commandkey{style1}] (0) at (0, -0.885) {$#1$};
    \node [style=\commandkey{style2}] (1) at (0, 0.885) {$#2$};
  \end{pgfonlayer}
  \begin{pgfonlayer}{edgelayer}
    \draw [\commandkey{estyle}] (0) to (1);
  \end{pgfonlayer}
\end{tikzpicture}\,}

\newcommand{\bra}[1]{\ensuremath{\left\langle #1 \right|}}
\newcommand{\ket}[1]{\ensuremath{\left|  #1 \right\rangle}}

\newcommand{\ketbra}[2]{\ensuremath{\ket{#1}\!\bra{#2}}}

\tikzstyle{cdiag}=[matrix of math nodes, row sep=3em, column sep=3em, text height=1.5ex, text depth=0.25ex,inner sep=0.5em]
\tikzstyle{arrow above}=[transform canvas={yshift=0.5ex}]
\tikzstyle{arrow below}=[transform canvas={yshift=-0.5ex}]

\tikzstyle{env}=[copoint,regular polygon rotate=0,minimum width=0.2cm, fill=black]

\tikzstyle{probs}=[shape=semicircle,fill=white,draw=black,shape border rotate=180,minimum width=1.2cm]

\tikzstyle{every picture}=[baseline=-0.25em,scale=0.5]
\tikzstyle{dotpic}=[] % for backwards-compatibility
\tikzstyle{diredges}=[every to/.style={diredge}]
\tikzstyle{math matrix}=[matrix of math nodes,left delimiter=(,right delimiter=),inner sep=2pt,column sep=1em,row sep=0.5em,nodes={inner sep=0pt},text height=1.5ex, text depth=0.25ex]

\tikzstyle{inline text}=[text height=1.5ex, text depth=0.25ex,yshift=0.5mm]
\tikzstyle{label}=[font=\footnotesize,text height=1.5ex, text depth=0.25ex,yshift=0.5mm]
\tikzstyle{left label}=[label,anchor=east,xshift=1.5mm]
\tikzstyle{right label}=[label,anchor=west,xshift=-1.5mm]

\tikzstyle{braceedge}=[decorate,decoration={brace,amplitude=2mm,raise=-1mm}]
\tikzstyle{small braceedge}=[decorate,decoration={brace,amplitude=1mm,raise=-1mm}]

\tikzstyle{doubled}=[line width=1.6pt] % set the line width for all doubled (quantum) maps/wires
\tikzstyle{boldedge}=[doubled,shorten <=-0.17mm,shorten >=-0.17mm]
\tikzstyle{boldedgegray}=[doubled,gray,shorten <=-0.17mm,shorten >=-0.17mm]
\tikzstyle{singleedgegray}=[gray]%,shorten <=-0.1mm,shorten >=-0.1mm]

\tikzstyle{semidoubled}=[line width=1.4pt] % set the line width for all doubled (quantum) maps/wires
\tikzstyle{semiboldedgegray}=[semidoubled,gray,shorten <=-0.17mm,shorten >=-0.17mm]

\tikzstyle{boxedge}=[semiboldedgegray]

\tikzstyle{boldedgedashed}=[very thick,dashed,shorten <=-0.17mm,shorten >=-0.17mm]
\tikzstyle{vboldedgedashed}=[doubled,dashed,shorten <=-0.17mm,shorten >=-0.17mm]
\tikzstyle{left hook arrow}=[left hook-latex]
\tikzstyle{right hook arrow}=[right hook-latex]
\tikzstyle{sembracket}=[line width=0.5pt,shorten <=-0.07mm,shorten >=-0.07mm]

\tikzstyle{causal edge}=[->,thick,gray]
\tikzstyle{causal nondir}=[thick,gray]
\tikzstyle{timeline}=[thick,gray, dashed]

\tikzstyle{cedge}=[<->,thick,gray!70!white]

\tikzstyle{empty diagram}=[draw=gray!40!white,dashed,shape=rectangle,minimum width=1cm,minimum height=1cm]
\tikzstyle{empty diagram small}=[draw=gray!50!white,dashed,shape=rectangle,minimum width=0.6cm,minimum height=0.5cm]

\tikzstyle{dot}=[inner sep=0mm,minimum width=2mm,minimum height=2mm,draw,shape=circle]  
\tikzstyle{Wsquare}=[white dot, shape=regular polygon, rounded corners=0.8 mm, minimum size=3.3 mm, regular polygon sides=3, outer sep=-0.2mm]
\tikzstyle{Wsquareadj}=[white dot, shape=regular polygon, rounded corners=0.8 mm, minimum size=3.3 mm, regular polygon sides=3, outer sep=-0.2mm, regular polygon rotate=180]
\tikzstyle{ddot}=[inner sep=0mm, doubled, minimum width=2.5mm,minimum height=2.5mm,draw,shape=circle]

\tikzstyle{black dot}=[dot,fill=black]
\tikzstyle{white dot}=[dot,fill=white,,text depth=-0.2mm]
\tikzstyle{white Wsquare}=[Wsquare,fill=white,,text depth=-0.2mm]
\tikzstyle{white Wsquareadj}=[Wsquareadj,fill=white,,text depth=-0.2mm]
\tikzstyle{green dot}=[white dot] % for backwards-compatibility
\tikzstyle{gray dot}=[dot,fill=gray!40!white,,text depth=-0.2mm]
\tikzstyle{red dot}=[gray dot] % for backwards-compatibility

\tikzstyle{black ddot}=[ddot,fill=black]
\tikzstyle{white ddot}=[ddot,fill=white]
\tikzstyle{gray ddot}=[ddot,fill=gray!40!white]

\tikzstyle{gray edge}=[gray!60!white]

\tikzstyle{small dot}=[inner sep=0.5mm,minimum width=0pt,minimum height=0pt,draw,shape=circle]

\tikzstyle{small black dot}=[small dot,fill=black]
\tikzstyle{small white dot}=[small dot,fill=white]
\tikzstyle{small gray dot}=[small dot,fill=gray!40!white]

\tikzstyle{causal dot}=[inner sep=0.4mm,minimum width=0pt,minimum height=0pt,draw=white,shape=circle,fill=gray!40!white]

\tikzstyle{phase dimensions}=[minimum size=5mm,font=\footnotesize,rectangle,rounded corners=2.5mm,inner sep=0.2mm,outer sep=-2mm]
\tikzstyle{dphase dimensions}=[minimum size=5mm,font=\footnotesize,rectangle,rounded corners=2.5mm,inner sep=0.2mm,outer sep=-2mm]
\tikzstyle{white phase dot}=[dot,fill=white,phase dimensions]
\tikzstyle{white phase ddot}=[ddot,fill=white,dphase dimensions]

\tikzstyle{white rect ddot}=[draw=black,fill=white,doubled,minimum size=5mm,font=\footnotesize,rectangle,rounded corners=2.5mm,inner sep=0.2mm]
\tikzstyle{gray rect ddot}=[draw=black,fill=gray!40!white,doubled,minimum size=6mm,font=\footnotesize,rectangle,rounded corners=3mm]

\tikzstyle{gray phase dot}=[dot,fill=gray!40!white,phase dimensions]
\tikzstyle{gray phase ddot}=[ddot,fill=gray!40!white,dphase dimensions]
\tikzstyle{grey phase dot}=[gray phase dot]
\tikzstyle{grey phase ddot}=[gray phase ddot]

\tikzstyle{small phase dimensions}=[minimum size=4mm,font=\tiny,rectangle,rounded corners=2mm,inner sep=0.2mm,outer sep=-2mm]
\tikzstyle{small dphase dimensions}=[minimum size=4mm,font=\tiny,rectangle,rounded corners=2mm,inner sep=0.2mm,outer sep=-2mm]

\tikzstyle{small gray phase dot}=[dot,fill=gray!40!white,small phase dimensions]
\tikzstyle{small gray phase ddot}=[ddot,fill=gray!40!white,small dphase dimensions]

\tikzstyle{small map}=[draw,shape=rectangle,minimum height=4mm,minimum width=4mm,fill=white]

\tikzstyle{cnot}=[fill=white,shape=circle,inner sep=-1.4pt]

\tikzstyle{asym hadamard}=[fill=white,draw,shape=NEbox,inner sep=0.6mm,font=\footnotesize,minimum height=4mm]
\tikzstyle{asym hadamard conj}=[fill=white,draw,shape=NWbox,inner sep=0.6mm,font=\footnotesize,minimum height=4mm]
\tikzstyle{asym hadamard dag}=[fill=white,draw,shape=SEbox,inner sep=0.6mm,font=\footnotesize,minimum height=4mm]

\tikzstyle{hadamard}=[fill=white,draw,inner sep=0.6mm,font=\footnotesize,minimum height=4mm,minimum width=4mm]
\tikzstyle{small hadamard}=[fill=white,draw,inner sep=0.6mm,minimum height=1.5mm,minimum width=1.5mm]
\tikzstyle{small hadamard rotate}=[small hadamard,rotate=45]
\tikzstyle{dhadamard}=[hadamard,doubled]
\tikzstyle{small dhadamard}=[small hadamard,doubled]
\tikzstyle{small dhadamard rotate}=[small hadamard rotate,doubled]
\tikzstyle{antipode}=[white dot,inner sep=0.3mm,font=\footnotesize]

\tikzstyle{scalar}=[diamond,draw,inner sep=0.5pt,font=\small]
\tikzstyle{dscalar}=[diamond,doubled, draw,inner sep=0.5pt,font=\small]

\tikzstyle{small box}=[rectangle,inline text,fill=white,draw,minimum height=5mm,yshift=-0.5mm,minimum width=5mm,font=\small]
\tikzstyle{small gray box}=[small box,fill=gray!30]
\tikzstyle{medium box}=[rectangle,inline text,fill=white,draw,minimum height=5mm,yshift=-0.5mm,minimum width=10mm,font=\small]
\tikzstyle{square box}=[small box] % for backwards-compatibility
\tikzstyle{medium gray box}=[small box,fill=gray!30]
\tikzstyle{semilarge box}=[rectangle,inline text,fill=white,draw,minimum height=5mm,yshift=-0.5mm,minimum width=12.5mm,font=\small]
\tikzstyle{large box}=[rectangle,inline text,fill=white,draw,minimum height=5mm,yshift=-0.5mm,minimum width=15mm,font=\small]
\tikzstyle{large gray box}=[small box,fill=gray!30]

\tikzstyle{Bayes box}=[rectangle,fill=black,draw, minimum height=3mm, minimum width=3mm]

\tikzstyle{gray square point}=[small box,fill=gray!50]

\tikzstyle{dphase box white}=[dhadamard]
\tikzstyle{dphase box gray}=[dhadamard,fill=gray!50!white]
\tikzstyle{phase box white}=[hadamard]
\tikzstyle{phase box gray}=[hadamard,fill=gray!50!white]

\tikzstyle{point}=[regular polygon,regular polygon sides=3,draw,scale=0.75,inner sep=-0.5pt,minimum width=9mm,fill=white,regular polygon rotate=180]
\tikzstyle{copoint}=[regular polygon,regular polygon sides=3,draw,scale=0.75,inner sep=-0.5pt,minimum width=9mm,fill=white]
\tikzstyle{dpoint}=[point,doubled]
\tikzstyle{dcopoint}=[copoint,doubled]

\tikzstyle{wide copoint}=[fill=white,draw,shape=isosceles triangle,shape border rotate=90,isosceles triangle stretches=true,inner sep=0pt,minimum width=1.5cm,minimum height=6.12mm]
\tikzstyle{wide point}=[fill=white,draw,shape=isosceles triangle,shape border rotate=-90,isosceles triangle stretches=true,inner sep=0pt,minimum width=1.5cm,minimum height=6.12mm,yshift=-0.0mm]
\tikzstyle{wide point plus}=[fill=white,draw,shape=isosceles triangle,shape border rotate=-90,isosceles triangle stretches=true,inner sep=0pt,minimum width=1.74cm,minimum height=7mm,yshift=-0.0mm]

\tikzstyle{wide dpoint}=[fill=white,doubled,draw,shape=isosceles triangle,shape border rotate=-90,isosceles triangle stretches=true,inner sep=0pt,minimum width=1.5cm,minimum height=6.12mm,yshift=-0.0mm]

\tikzstyle{tinypoint}=[regular polygon,regular polygon sides=3,draw,scale=0.55,inner sep=-0.15pt,minimum width=6mm,fill=white,regular polygon rotate=180] 

\tikzstyle{white point}=[point]
\tikzstyle{white dpoint}=[dpoint]
\tikzstyle{green point}=[white point] % for backwards-compatibility
\tikzstyle{white copoint}=[copoint]
\tikzstyle{gray point}=[point,fill=gray!40!white]
\tikzstyle{gray dpoint}=[gray point,doubled]
\tikzstyle{red point}=[gray point] % for backwards-compatibility
\tikzstyle{gray copoint}=[copoint,fill=gray!40!white]
\tikzstyle{gray dcopoint}=[gray copoint,doubled]

\tikzstyle{white point guide}=[regular polygon,regular polygon sides=3,font=\scriptsize,draw,scale=0.65,inner sep=-0.5pt,minimum width=9mm,fill=white,regular polygon rotate=180]

\tikzstyle{black point}=[point,fill=black,font=\color{white}]
\tikzstyle{black copoint}=[copoint,fill=black,font=\color{white}]

\tikzstyle{tiny gray point}=[tinypoint,fill=gray!40!white]

\tikzstyle{diredge}=[->]
\tikzstyle{ddiredge}=[<->]
\tikzstyle{rdiredge}=[<-]
\tikzstyle{thickdiredge}=[->, very thick]
\tikzstyle{pointer edge}=[->,very thick,gray]
\tikzstyle{pointer edge part}=[very thick,gray]
\tikzstyle{dashed edge}=[dashed]
\tikzstyle{thick dashed edge}=[very thick,dashed]
\tikzstyle{thick gray dashed edge}=[thick dashed edge,gray!40]
\tikzstyle{thick map edge}=[very thick,|->]

\makeatletter
\newcommand{\boxshape}[3]{%
\pgfdeclareshape{#1}{
\inheritsavedanchors[from=rectangle] % this is nearly a rectangle
\inheritanchorborder[from=rectangle]
\inheritanchor[from=rectangle]{center}
\inheritanchor[from=rectangle]{north}
\inheritanchor[from=rectangle]{south}
\inheritanchor[from=rectangle]{west}
\inheritanchor[from=rectangle]{east}
\backgroundpath{% this is new
\southwest \pgf@xa=\pgf@x \pgf@ya=\pgf@y
\northeast \pgf@xb=\pgf@x \pgf@yb=\pgf@y

\@tempdima=#2
\@tempdimb=#3

\pgfpathmoveto{\pgfpoint{\pgf@xa - 5pt + \@tempdima}{\pgf@ya}}
\pgfpathlineto{\pgfpoint{\pgf@xa - 5pt - \@tempdima}{\pgf@yb}}
\pgfpathlineto{\pgfpoint{\pgf@xb + 5pt + \@tempdimb}{\pgf@yb}}
\pgfpathlineto{\pgfpoint{\pgf@xb + 5pt - \@tempdimb}{\pgf@ya}}
\pgfpathlineto{\pgfpoint{\pgf@xa - 5pt + \@tempdima}{\pgf@ya}}
\pgfpathclose
}
}}

\boxshape{NEbox}{0pt}{5pt}
\boxshape{SEbox}{0pt}{-5pt}
\boxshape{NWbox}{5pt}{0pt}
\boxshape{SWbox}{-5pt}{0pt}
\boxshape{EBox}{-3pt}{3pt}
\boxshape{WBox}{3pt}{-3pt}
\makeatother

\tikzstyle{cloud}=[shape=cloud,draw,minimum width=1.5cm,minimum height=1.5cm]

\tikzstyle{map}=[draw,shape=NEbox,inner sep=2pt,minimum height=6mm,fill=white]
\tikzstyle{dashedmap}=[draw,dashed,shape=NEbox,inner sep=2pt,minimum height=6mm,fill=white]
\tikzstyle{mapdag}=[draw,shape=SEbox,inner sep=2pt,minimum height=6mm,fill=white]
\tikzstyle{mapadj}=[draw,shape=SEbox,inner sep=2pt,minimum height=6mm,fill=white]
\tikzstyle{maptrans}=[draw,shape=SWbox,inner sep=2pt,minimum height=6mm,fill=white]
\tikzstyle{mapconj}=[draw,shape=NWbox,inner sep=2pt,minimum height=6mm,fill=white]

\tikzstyle{medium map}=[draw,shape=NEbox,inner sep=2pt,minimum height=6mm,fill=white,minimum width=7mm]
\tikzstyle{medium map dag}=[draw,shape=SEbox,inner sep=2pt,minimum height=6mm,fill=white,minimum width=7mm]
\tikzstyle{medium map adj}=[draw,shape=SEbox,inner sep=2pt,minimum height=6mm,fill=white,minimum width=7mm]
\tikzstyle{medium map trans}=[draw,shape=SWbox,inner sep=2pt,minimum height=6mm,fill=white,minimum width=7mm]
\tikzstyle{medium map conj}=[draw,shape=NWbox,inner sep=2pt,minimum height=6mm,fill=white,minimum width=7mm]
\tikzstyle{semilarge map}=[draw,shape=NEbox,inner sep=2pt,minimum height=6mm,fill=white,minimum width=9.5mm]
\tikzstyle{semilarge map trans}=[draw,shape=SWbox,inner sep=2pt,minimum height=6mm,fill=white,minimum width=9.5mm]
\tikzstyle{semilarge map adj}=[draw,shape=SEbox,inner sep=2pt,minimum height=6mm,fill=white,minimum width=9.5mm]
\tikzstyle{semilarge map dag}=[draw,shape=SEbox,inner sep=2pt,minimum height=6mm,fill=white,minimum width=9.5mm]
\tikzstyle{semilarge map conj}=[draw,shape=NWbox,inner sep=2pt,minimum height=6mm,fill=white,minimum width=9.5mm]
\tikzstyle{large map}=[draw,shape=NEbox,inner sep=2pt,minimum height=6mm,fill=white,minimum width=12mm]
\tikzstyle{large map conj}=[draw,shape=NWbox,inner sep=2pt,minimum height=6mm,fill=white,minimum width=12mm]
\tikzstyle{very large map}=[draw,shape=NEbox,inner sep=2pt,minimum height=6mm,fill=white,minimum width=17mm]

\tikzstyle{medium dmap}=[draw,doubled,shape=NEbox,inner sep=2pt,minimum height=6mm,fill=white,minimum width=7mm]
\tikzstyle{medium dmap dag}=[draw,doubled,shape=SEbox,inner sep=2pt,minimum height=6mm,fill=white,minimum width=7mm]
\tikzstyle{medium dmap adj}=[draw,doubled,shape=SEbox,inner sep=2pt,minimum height=6mm,fill=white,minimum width=7mm]
\tikzstyle{medium dmap trans}=[draw,doubled,shape=SWbox,inner sep=2pt,minimum height=6mm,fill=white,minimum width=7mm]
\tikzstyle{medium dmap conj}=[draw,doubled,shape=NWbox,inner sep=2pt,minimum height=6mm,fill=white,minimum width=7mm]
\tikzstyle{semilarge dmap}=[draw,doubled,shape=NEbox,inner sep=2pt,minimum height=6mm,fill=white,minimum width=9.5mm]
\tikzstyle{semilarge dmap trans}=[draw,doubled,shape=SWbox,inner sep=2pt,minimum height=6mm,fill=white,minimum width=9.5mm]
\tikzstyle{semilarge dmap adj}=[draw,doubled,shape=SEbox,inner sep=2pt,minimum height=6mm,fill=white,minimum width=9.5mm]
\tikzstyle{semilarge dmap dag}=[draw,doubled,shape=SEbox,inner sep=2pt,minimum height=6mm,fill=white,minimum width=9.5mm]
\tikzstyle{semilarge dmap conj}=[draw,doubled,shape=NWbox,inner sep=2pt,minimum height=6mm,fill=white,minimum width=9.5mm]
\tikzstyle{large dmap}=[draw,doubled,shape=NEbox,inner sep=2pt,minimum height=6mm,fill=white,minimum width=12mm]
\tikzstyle{large dmap conj}=[draw,doubled,shape=NWbox,inner sep=2pt,minimum height=6mm,fill=white,minimum width=12mm]
\tikzstyle{large dmap trans}=[draw,doubled,shape=SWbox,inner sep=2pt,minimum height=6mm,fill=white,minimum width=12mm]
\tikzstyle{large dmap adj}=[draw,doubled,shape=SEbox,inner sep=2pt,minimum height=6mm,fill=white,minimum width=12mm]
\tikzstyle{large dmap dag}=[draw,doubled,shape=SEbox,inner sep=2pt,minimum height=6mm,fill=white,minimum width=12mm]
\tikzstyle{very large dmap}=[draw,doubled,shape=NEbox,inner sep=2pt,minimum height=6mm,fill=white,minimum width=19.5mm]

\tikzstyle{muxbox}=[draw,shape=rectangle,minimum height=3mm,minimum width=3mm,fill=white]
\tikzstyle{dmuxbox}=[muxbox,doubled]

\tikzstyle{box}=[draw,shape=rectangle,inner sep=2pt,minimum height=6mm,minimum width=6mm,fill=white]
\tikzstyle{dbox}=[draw,doubled,shape=rectangle,inner sep=2pt,minimum height=6mm,minimum width=6mm,fill=white]
\tikzstyle{dmap}=[draw,doubled,shape=NEbox,inner sep=2pt,minimum height=6mm,fill=white]
\tikzstyle{dmapdag}=[draw,doubled,shape=SEbox,inner sep=2pt,minimum height=6mm,fill=white]
\tikzstyle{dmapadj}=[draw,doubled,shape=SEbox,inner sep=2pt,minimum height=6mm,fill=white]
\tikzstyle{dmaptrans}=[draw,doubled,shape=SWbox,inner sep=2pt,minimum height=6mm,fill=white]
\tikzstyle{dmapconj}=[draw,doubled,shape=NWbox,inner sep=2pt,minimum height=6mm,fill=white]

\tikzstyle{ddmap}=[draw,doubled,dashed,shape=NEbox,inner sep=2pt,minimum height=6mm,fill=white]
\tikzstyle{ddmapdag}=[draw,doubled,dashed,shape=SEbox,inner sep=2pt,minimum height=6mm,fill=white]
\tikzstyle{ddmapadj}=[draw,doubled,dashed,shape=SEbox,inner sep=2pt,minimum height=6mm,fill=white]
\tikzstyle{ddmaptrans}=[draw,doubled,dashed,shape=SWbox,inner sep=2pt,minimum height=6mm,fill=white]
\tikzstyle{ddmapconj}=[draw,doubled,dashed,shape=NWbox,inner sep=2pt,minimum height=6mm,fill=white]

\boxshape{sNEbox}{0pt}{3pt}
\boxshape{sSEbox}{0pt}{-3pt}
\boxshape{sNWbox}{3pt}{0pt}
\boxshape{sSWbox}{-3pt}{0pt}
\tikzstyle{smap}=[draw,shape=sNEbox,fill=white]
\tikzstyle{smapdag}=[draw,shape=sSEbox,fill=white]
\tikzstyle{smapadj}=[draw,shape=sSEbox,fill=white]
\tikzstyle{smaptrans}=[draw,shape=sSWbox,fill=white]
\tikzstyle{smapconj}=[draw,shape=sNWbox,fill=white]

\tikzstyle{dsmap}=[draw,dashed,shape=sNEbox,fill=white]
\tikzstyle{dsmapdag}=[draw,dashed,shape=sSEbox,fill=white]
\tikzstyle{dsmaptrans}=[draw,dashed,shape=sSWbox,fill=white]
\tikzstyle{dsmapconj}=[draw,dashed,shape=sNWbox,fill=white]

\boxshape{mNEbox}{0pt}{10pt}
\boxshape{mSEbox}{0pt}{-10pt}
\boxshape{mNWbox}{10pt}{0pt}
\boxshape{mSWbox}{-10pt}{0pt}
\tikzstyle{mmap}=[draw,shape=mNEbox]
\tikzstyle{mmapdag}=[draw,shape=mSEbox]
\tikzstyle{mmaptrans}=[draw,shape=mSWbox]
\tikzstyle{mmapconj}=[draw,shape=mNWbox]

\tikzstyle{mmapgray}=[draw,fill=gray!40!white,shape=mNEbox]
\tikzstyle{smapgray}=[draw,fill=gray!40!white,shape=sNEbox]

\makeatletter

\pgfdeclareshape{cornerpoint}{
\inheritsavedanchors[from=rectangle] % this is nearly a rectangle
\inheritanchorborder[from=rectangle]
\inheritanchor[from=rectangle]{center}
\inheritanchor[from=rectangle]{north}
\inheritanchor[from=rectangle]{south}
\inheritanchor[from=rectangle]{west}
\inheritanchor[from=rectangle]{east}
\backgroundpath{% this is new
\southwest \pgf@xa=\pgf@x \pgf@ya=\pgf@y
\northeast \pgf@xb=\pgf@x \pgf@yb=\pgf@y

\pgfmathsetmacro{\pgf@shorten@left}{\pgfkeysvalueof{/tikz/shorten left}}
\pgfmathsetmacro{\pgf@shorten@right}{\pgfkeysvalueof{/tikz/shorten right}}

\pgfpathmoveto{\pgfpoint{0.5 * (\pgf@xa + \pgf@xb)}{\pgf@ya - 5pt}}
\pgfpathlineto{\pgfpoint{\pgf@xa - 8pt + \pgf@shorten@left}{\pgf@yb - 1.5 * \pgf@shorten@left}}
\pgfpathlineto{\pgfpoint{\pgf@xa - 8pt + \pgf@shorten@left}{\pgf@yb}}
\pgfpathlineto{\pgfpoint{\pgf@xb + 8pt - \pgf@shorten@right}{\pgf@yb}}
\pgfpathlineto{\pgfpoint{\pgf@xb + 8pt - \pgf@shorten@right}{\pgf@yb - 1.5 * \pgf@shorten@right}}
\pgfpathclose
}
}

\pgfdeclareshape{cornercopoint}{
\inheritsavedanchors[from=rectangle] % this is nearly a rectangle
\inheritanchorborder[from=rectangle]
\inheritanchor[from=rectangle]{center}
\inheritanchor[from=rectangle]{north}
\inheritanchor[from=rectangle]{south}
\inheritanchor[from=rectangle]{west}
\inheritanchor[from=rectangle]{east}
\backgroundpath{% this is new
\southwest \pgf@xa=\pgf@x \pgf@ya=\pgf@y
\northeast \pgf@xb=\pgf@x \pgf@yb=\pgf@y

\pgfmathsetmacro{\pgf@shorten@left}{\pgfkeysvalueof{/tikz/shorten left}}
\pgfmathsetmacro{\pgf@shorten@right}{\pgfkeysvalueof{/tikz/shorten right}}

\pgfpathmoveto{\pgfpoint{0.5 * (\pgf@xa + \pgf@xb)}{\pgf@yb + 5pt}}
\pgfpathlineto{\pgfpoint{\pgf@xa - 8pt + \pgf@shorten@left}{\pgf@ya + 1.5 * \pgf@shorten@left}}
\pgfpathlineto{\pgfpoint{\pgf@xa - 8pt + \pgf@shorten@left}{\pgf@ya}}
\pgfpathlineto{\pgfpoint{\pgf@xb + 8pt - \pgf@shorten@right}{\pgf@ya}}
\pgfpathlineto{\pgfpoint{\pgf@xb + 8pt - \pgf@shorten@right}{\pgf@ya + 1.5 * \pgf@shorten@right}}
\pgfpathclose
}
}

\makeatother

\pgfkeyssetvalue{/tikz/shorten left}{0pt}
\pgfkeyssetvalue{/tikz/shorten right}{0pt}

\tikzstyle{kpoint common}=[draw,fill=white,inner sep=1pt,minimum height=4mm]
\tikzstyle{kpoint sc}=[shape=cornerpoint,kpoint common]
\tikzstyle{kpoint adjoint sc}=[shape=cornercopoint,kpoint common]
\tikzstyle{kpoint}=[shape=cornerpoint,shorten left=5pt,kpoint common]
\tikzstyle{kpoint adjoint}=[shape=cornercopoint,shorten left=5pt,kpoint common]
\tikzstyle{kpoint conjugate}=[shape=cornerpoint,shorten right=5pt,kpoint common]
\tikzstyle{kpoint transpose}=[shape=cornercopoint,shorten right=5pt,kpoint common]
\tikzstyle{kpoint symm}=[shape=cornerpoint,shorten left=5pt,shorten right=5pt,kpoint common]

\tikzstyle{black kpoint}=[shape=cornerpoint,shorten left=5pt,kpoint common,fill=black,font=\color{white}]
\tikzstyle{black kpoint adjoint}=[shape=cornercopoint,shorten left=5pt,kpoint common,fill=black,font=\color{white}]
\tikzstyle{black kpointadj}=[shape=cornercopoint,shorten left=5pt,kpoint common,fill=black,font=\color{white}]

\tikzstyle{black dkpoint}=[shape=cornerpoint,shorten left=5pt,kpoint common,fill=black, doubled,font=\color{white}]
\tikzstyle{black dkpoint adjoint}=[shape=cornercopoint,shorten left=5pt,kpoint common,fill=black, doubled,font=\color{white}]
\tikzstyle{black dkpointadj}=[shape=cornercopoint,shorten left=5pt,kpoint common,fill=black, doubled,font=\color{white}] 

\tikzstyle{kpointdag}=[kpoint adjoint]
\tikzstyle{kpointadj}=[kpoint adjoint]
\tikzstyle{kpointconj}=[kpoint conjugate]
\tikzstyle{kpointtrans}=[kpoint transpose]

\tikzstyle{big kpoint}=[kpoint, minimum width=1.2 cm, minimum height=8mm, inner sep=4pt, text depth=3mm]

\tikzstyle{wide kpoint}=[kpoint, minimum width=1 cm, inner sep=2pt]%, text depth=-0.7 mm]
\tikzstyle{wide kpointdag}=[kpointdag, minimum width=1 cm, inner sep=2pt]%, text depth=0.7 mm]
\tikzstyle{wide kpointconj}=[kpointconj, minimum width=1 cm, inner sep=2pt]%, text depth=-0.7 mm]
\tikzstyle{wide kpointtrans}=[kpointtrans, minimum width=1 cm, inner sep=2pt]%, text depth=0.7 mm]

\tikzstyle{gray kpoint}=[kpoint,fill=gray!50!white]
\tikzstyle{gray kpointdag}=[kpointdag,fill=gray!50!white]
\tikzstyle{gray kpointadj}=[kpointadj,fill=gray!50!white]
\tikzstyle{gray kpointconj}=[kpointconj,fill=gray!50!white]
\tikzstyle{gray kpointtrans}=[kpointtrans,fill=gray!50!white]

\tikzstyle{gray dkpoint}=[kpoint,fill=gray!50!white,doubled]
\tikzstyle{gray dkpointdag}=[kpointdag,fill=gray!50!white,doubled]
\tikzstyle{gray dkpointadj}=[kpointadj,fill=gray!50!white,doubled]
\tikzstyle{gray dkpointconj}=[kpointconj,fill=gray!50!white,doubled]
\tikzstyle{gray dkpointtrans}=[kpointtrans,fill=gray!50!white,doubled]

\tikzstyle{white label}=[draw,fill=white,rectangle,inner sep=0.7 mm]
\tikzstyle{gray label}=[draw,fill=gray!50!white,rectangle,inner sep=0.7 mm]
\tikzstyle{black label}=[draw,fill=black,rectangle,inner sep=0.7 mm]

\tikzstyle{dkpoint}=[kpoint,doubled]
\tikzstyle{wide dkpoint}=[wide kpoint,doubled]
\tikzstyle{dkpointdag}=[kpoint adjoint,doubled]
\tikzstyle{wide dkpointdag}=[wide kpointdag,doubled]
\tikzstyle{dkcopoint}=[kpoint adjoint,doubled]
\tikzstyle{dkpointadj}=[kpoint adjoint,doubled]
\tikzstyle{dkpointconj}=[kpoint conjugate,doubled]
\tikzstyle{dkpointtrans}=[kpoint transpose,doubled]

\tikzstyle{kscalar}=[kpoint common, shape=EBox, inner xsep=-1pt, inner ysep=3pt,font=\small]
\tikzstyle{kscalarconj}=[kpoint common, shape=WBox, inner xsep=-1pt, inner ysep=3pt,font=\small]

\tikzstyle{spekpoint}=[kpoint sc,minimum height=5mm,inner sep=3pt]
\tikzstyle{spekcopoint}=[kpoint adjoint sc,minimum height=5mm,inner sep=3pt]

\tikzstyle{dspekpoint}=[spekpoint,doubled]
\tikzstyle{dspekcopoint}=[spekcopoint,doubled]

 \tikzstyle{upground}=[circuit ee IEC,thick,ground,rotate=90,scale=2.3]
 \tikzstyle{downground}=[circuit ee IEC,thick,ground,rotate=-90,scale=2.3]
 \tikzstyle{bigground}=[regular polygon,regular polygon sides=3,draw=gray,scale=0.50,inner sep=-0.5pt,minimum width=10mm,fill=gray]
\tikzstyle{arrs}=[-latex,font=\small,auto]
\tikzstyle{arrow plain}=[arrs]
\tikzstyle{arrow dashed}=[dashed,arrs]
\tikzstyle{arrow bold}=[very thick,arrs]
\tikzstyle{arrow hide}=[draw=white!0,-]
\tikzstyle{arrow reverse}=[latex-]
\tikzstyle{cdnode}=[]

\let\olddagger\dagger
\renewcommand{\dagger}{\ensuremath{\olddagger}\xspace}

\theoremstyle{definition}
\newtheorem{theorem}{Theorem}[section]
\newtheorem{corollary}[theorem]{Corollary}

\newtheorem{definition}[theorem]{Definition}

\newtheorem{example*}[theorem]{Example*}
\newtheorem{examples*}[theorem]{Examples*}

\newtheorem{remark*}[theorem]{Remark*}

\usepackage[color]{changebar}

\usepackage{color}
\def\bR{\begin{color}{red}} 
\def\bB{\begin{color}{blue}}
\def\bM{\begin{color}{magenta}}
\def\bC{\begin{color}{cyan}}
\def\bW{\begin{color}{white}}
\def\bBl{\begin{color}{black}} 
\def\bG{\begin{color}{green}}
\def\bY{\begin{color}{yellow}}
\def\e{\end{color}\xspace}
\newcommand{\bit}{\begin{itemize}}
\newcommand{\eit}{\end{itemize}\par\noindent}
\newcommand{\ben}{\begin{enumerate}}
\newcommand{\een}{\end{enumerate}\par\noindent}
\newcommand{\beq}{\begin{equation}}
\newcommand{\eeq}{\end{equation}\par\noindent}
\newcommand{\beqa}{\begin{eqnarray*}}
\newcommand{\eeqa}{\end{eqnarray*}\par\noindent}
\newcommand{\beqn}{\begin{eqnarray}}
\newcommand{\eeqn}{\end{eqnarray}\par\noindent}
\newcommand{\SOCt}{\ensuremath{\textrm{SOC}_2}\xspace}
\usepackage{braket}

\title{Picturing Indefinite Causal Structure}
\author{Aleks Kissinger and Sander Uijlen \\
{\small Institute for Computing and Information Sciences, Radboud University} \\
{\small\tt \{aleks,suijlen\}@cs.ru.nl}}
\def\titlerunning{Picturing Indefinite Causal Structure}
\def\authorrunning{Aleks Kissinger \& Sander Uijlen}

\begin{document}

\maketitle

\begin{abstract}
  Following on from the notion of (first-order) causality, which generalises the notion of being trace-preserving from CP-maps to abstract processes, we give a characterization for the most general kind of map which sends causal processes to causal processes. These new, second-order causal processes enable us to treat the input processes as `local laboratories' whose causal ordering needs not be fixed in advance. Using this characterization, we give a fully-diagrammatic proof of a non-trivial theorem: namely that being causality-preserving on separable processes implies being `completely' causality-preserving. That is, causality is preserved even when the `local laboratories' are allowed to have ancilla systems. An immediate consequence is that preserving causality is separable processes is equivalence to preserving causality for strongly non-signalling (a.k.a. localizable) processes.
\end{abstract}

% \section{Introduction}

% Recently, frameworks have been proposed from studying quantum correlations with no fixed causal ordering. In addition to being of interest in their own right, such scenarios are likely to occur in a quantum theory of gravity. While they have thus far only been studied in the context of quantum theory, many interesting properties already appear at the abstract level of categories where it makes sense to talk about `causal processes'. That is, categoies with a distinguished `discarding' map, which is preserved by the causal processes. Starting from this setting, we

\section{Causality and non-signalling}

Throughout this extended abstract, we will work in a self-dual compact closed category $\mathcal C$, that is, a symmetric monoidal category which has for every object a pair of morphisms:
\[
\eta_{A} : I \to A \otimes A
\qquad\qquad
\epsilon_{A} : A \otimes A \to I
\]
which we refer to as \textit{cups} and \textit{caps} respectively, satisfying the following `yanking' identities:
\[
(\epsilon_A \otimes 1_A) \circ (1_A \otimes \eta_A) = 1_A
\qquad
\gamma_A \circ \eta_A = \eta_A
\qquad
\epsilon_A \circ \gamma_A  = \epsilon_A
\]
where $\gamma_A : A \otimes A \to A \otimes A$ is the symmetry natural isomorphism. Furthermore, we will adopt string diagram notion for depicting compositions of morphisms (see e.g. \cite{SelingerSurvey}). Using this notion, cups and caps resemble their namesakes:
\[ \eta_A\ :=\ \begin{tikzpicture}
	\begin{pgfonlayer}{nodelayer}
		\node [style=none] (0) at (-0.75, 0.25) {};
		\node [style=none] (1) at (0.75, 0.25) {};
	\end{pgfonlayer}
	\begin{pgfonlayer}{edgelayer}
		\draw [in=-90, out=-90, looseness=1.75] (0.center) to (1.center);
	\end{pgfonlayer}
\end{tikzpicture} \qquad\qquad \epsilon_A\ :=\ \begin{tikzpicture}
	\begin{pgfonlayer}{nodelayer}
		\node [style=none] (0) at (-0.75, -0.5) {};
		\node [style=none] (1) at (0.75, -0.5) {};
	\end{pgfonlayer}
	\begin{pgfonlayer}{edgelayer}
		\draw [in=90, out=90, looseness=1.75] (0.center) to (1.center);  
	\end{pgfonlayer}
\end{tikzpicture} \]
and hence the equations above become:
\[
\begin{tikzpicture}
	\begin{pgfonlayer}{nodelayer}
		\node [style=none] (0) at (-5.5, -1.25) {};
		\node [style=none] (1) at (0, -1.25) {};
		\node [style=none] (2) at (-4, 0) {};
		\node [style=none] (3) at (-2.5, 0) {};
		\node [style=none] (4) at (-1.25, 0) {$=$};
		\node [style=none] (5) at (-5.5, 0) {};
		\node [style=none] (6) at (-4, 0) {};
		\node [style=none] (7) at (-2.5, 1.25) {};
		\node [style=none] (8) at (0, 1.25) {};
	\end{pgfonlayer}
	\begin{pgfonlayer}{edgelayer}
		\draw [in=-90, out=-90, looseness=1.75] (2.center) to (3.center);
		\draw (3.center) to (7.center);
		\draw [in=90, out=90, looseness=1.75] (5.center) to (6.center);
		\draw (0.center) to (5.center);
		\draw [style=swap] (8.center) to (1.center);
	\end{pgfonlayer}
\end{tikzpicture}
\qquad\qquad
\begin{tikzpicture}
	\begin{pgfonlayer}{nodelayer}
		\node [style=none] (0) at (-0.75, 1) {};
		\node [style=none] (1) at (0.75, 1) {};
		\node [style=none] (2) at (-0.75, -0.25) {};
		\node [style=none] (3) at (0.75, -0.25) {};
	\end{pgfonlayer}
	\begin{pgfonlayer}{edgelayer}
		\draw [in=90, out=-90, looseness=1.00] (1.center) to (2.center);
		\draw [in=-90, out=-90, looseness=1.75] (2.center) to (3.center);
		\draw [in=-90, out=90, looseness=1.00] (3.center) to (0.center); 
	\end{pgfonlayer}
\end{tikzpicture}  \ \ = \ \ \begin{tikzpicture}
	\begin{pgfonlayer}{nodelayer}
		\node [style=none] (0) at (-0.75, 0.25) {};
		\node [style=none] (1) at (0.75, 0.25) {};
	\end{pgfonlayer}
	\begin{pgfonlayer}{edgelayer}
		\draw [in=-90, out=-90, looseness=1.75] (0.center) to (1.center);
	\end{pgfonlayer}
\end{tikzpicture}
\qquad\qquad
\begin{tikzpicture}
	\begin{pgfonlayer}{nodelayer}
		\node [style=none] (0) at (0.75, -1) {};
		\node [style=none] (1) at (0.75, 0.25) {};
		\node [style=none] (2) at (-0.75, -1) {};
		\node [style=none] (3) at (-0.75, 0.25) {};
	\end{pgfonlayer}
	\begin{pgfonlayer}{edgelayer}
		\draw [in=-90, out=90, looseness=1.00] (0.center) to (3.center);
		\draw [in=90, out=90, looseness=1.75] (3.center) to (1.center);
		\draw [in=90, out=-90, looseness=1.00] (1.center) to (2.center); 
	\end{pgfonlayer}
\end{tikzpicture}  \ \ = \ \ \begin{tikzpicture}
	\begin{pgfonlayer}{nodelayer}
		\node [style=none] (0) at (-0.75, -0.5) {};
		\node [style=none] (1) at (0.75, -0.5) {};
	\end{pgfonlayer}
	\begin{pgfonlayer}{edgelayer}
		\draw [in=90, out=90, looseness=1.75] (0.center) to (1.center);  
	\end{pgfonlayer}
\end{tikzpicture}
\]
Note that the monoidal unit $I$ is depicted as empty space. Throughout the paper, we will think of morphisms in this category as physical processes of some kind, hence we adopt `process-theoretic' language. Namely, we refer to objects as \textit{systems} and morphisms as \textit{processes}. Furthermore, we give special names to processes from and to the trivial system:
\[ \textit{states}\ :=\ \pointmap\psi \qquad\qquad \textit{effects}\ :=\ \copointmap\pi \]

In addition to the compact closed structure, we also assume $\mathcal C$ has a distinguished effect $d_A : A \to I$ for every system $A$ called \textit{discarding}. This is pictured as:
\[ d_A := \sdiscard \]
and is compatible with $\otimes$ and $I$ as follows:
\[ \begin{tikzpicture}
	\begin{pgfonlayer}{nodelayer}
		\node [style=upground] (0) at (-1, 0.5) {};
		\node [style=none] (1) at (-1.25, 0.25) {};
		\node [style=none] (2) at (-0.75, 0.25) {};
		\node [style=none] (3) at (-1.25, -0.5) {};
		\node [style=none] (4) at (-0.75, -0.5) {};
		\node [style=none] (5) at (0.5, 0) {$=$};
		\node [style=none] (6) at (2, -0.5) {};
		\node [style=none] (7) at (2, 0.25) {};
		\node [style=none] (8) at (3.5, 0.25) {};
		\node [style=upground] (9) at (2, 0.5) {};
		\node [style=none] (10) at (3.5, -0.5) {};
		\node [style=upground] (11) at (3.5, 0.5) {};
		\node [style=none] (12) at (-3.5, 0) {$d_{A\otimes B} \ =$};
		\node [style=none] (13) at (6.25, 0) {$=\ d_A \otimes d_B$};
	\end{pgfonlayer}
	\begin{pgfonlayer}{edgelayer}
		\draw (3.center) to (1.center);
		\draw (4.center) to (2.center);
		\draw (6.center) to (7.center);
		\draw (10.center) to (8.center);
	\end{pgfonlayer}
\end{tikzpicture} \qquad\qquad \begin{tikzpicture}
	\begin{pgfonlayer}{nodelayer}
		\node [style=none] (0) at (-2.5, 0) {$d_{I} \ =$};
		\node [style=none] (1) at (2.5, 0) {$=\ 1_I$};
		\node [style=empty diagram] (2) at (0, 0) {};
	\end{pgfonlayer}
\end{tikzpicture} \]
The utility of the discarding process is it enables us to define \textit{causality}, following \cite{CRCaucat,Chiri1}:

\begin{definition}
  A process $\Psi : A \to B$ is called \textit{causal} if $d_B \circ \Psi = d_A$, or pictorially:
  \[ \sboxdiscard{\Psi}\ =\ \sdiscard \]
\end{definition}

\noindent The motto for causal processes is therefore:
\begin{center}
  \it If we discard the output of a process, it doesn't\\ matter which process happened.
\end{center}

In the category whose objects are quantum state spaces and whose morphisms are CP-maps, causality corresponds to being a trace-preserving CP-map, i.e. a \textit{quantum channel}.

The utility of causality is that it enables us to use diagrams to represent the causal relationships between processes~\cite{CRCaucat}. For example, if we wish to express that Alice can signal to Bob (but not vice-versa!), we can require that a causal process:
\[ \Phi : A_1 \otimes A_2 \to B_1 \otimes B_2 \]
factorises as:
\ctikzfig{factorAB}
where $\Psi_A$ and $\Psi_B$ are also causal. Following \cite{Cnonsig}, we see from this factorisation, that it is indeed impossible for Bob to signal Alice. Indeed, if we discard Bob's output (to which Alice does not have access), the whole process disconnects:
\ctikzfig{factorAB-sep}
We say such a process is \textit{non-signalling} from $B$ to $A$, and write $A \preceq B$. Similarly, a process is non-signalling from $A$ to $B$ if it factorises as:
\ctikzfig{factorBA}
and we say it is simply \textit{non-signalling} if it admits both factorisations.

A typical example of a non-signalling process is a Bell-type scenario. That is, Alice and Bob share some bipartite state, to which they perform local operations:
\begin{equation}\label{eq:factor-bell}
  \begin{tikzpicture}
	\begin{pgfonlayer}{nodelayer}
		\node [style=medium box] (0) at (-2, 0) {$\Phi$};
		\node [style=none] (1) at (-2.75, -0.25) {};
		\node [style=none] (2) at (-1.25, -0.25) {};
		\node [style=none] (3) at (-2.75, -1.5) {};
		\node [style=none] (4) at (-1.25, -1.5) {};
		\node [style=none] (5) at (-2.75, 0.25) {};
		\node [style=none] (6) at (-1.25, 0.25) {};
		\node [style=none] (7) at (-2.75, 1.5) {};
		\node [style=none] (8) at (-1.25, 1.5) {};
		\node [style=none] (9) at (0, 0) {$=$};
		\node [style=none] (10) at (1.5, -2) {};
		\node [style=none] (11) at (1.75, 2) {};
		\node [style=none] (12) at (1.75, 1.25) {};
		\node [style=none] (13) at (1.5, 0.75) {};
		\node [style=small box] (14) at (1.75, 1) {$\Psi_A$};
		\node [style=small box] (15) at (4.75, 1) {$\Psi_B$};
		\node [style=none] (16) at (2.75, -0.75) {};
		\node [style=none] (17) at (2, 0.75) {};
		\node [style=none] (18) at (3.75, -0.75) {};
		\node [style=none] (19) at (5, 0.75) {};
		\node [style=none] (20) at (5, -2) {};
		\node [style=none] (21) at (4.75, 1.25) {};
		\node [style=none] (22) at (4.5, 0.75) {};
		\node [style=none] (23) at (4.75, 2) {};
		\node [style=kpoint sc, minimum width=5 mm] (24) at (3.25, -1) {$\rho$};
	\end{pgfonlayer}
	\begin{pgfonlayer}{edgelayer}
		\draw (3.center) to (1.center);
		\draw (5.center) to (7.center);
		\draw (4.center) to (2.center);
		\draw (6.center) to (8.center);
		\draw (10.center) to (13.center);
		\draw (12.center) to (11.center);
		\draw [in=-90, out=90, looseness=1.00] (16.center) to (17.center);
		\draw (20.center) to (19.center);
		\draw (21.center) to (23.center);
		\draw [in=-90, out=90, looseness=1.00] (18.center) to (22.center);
	\end{pgfonlayer}
\end{tikzpicture}
\end{equation}
This clearly admits the two factorisations for $A \preceq B$ and $B \preceq A$:
\ctikzfig{factor-bell2}
so one might ask if in fact \textit{all} non-signalling processes arise this way. In quantum theory, surprisingly the answer is no.

%However, non-signalling processes which are not of this form are pretty hard to come by~\cite{???}.

\begin{definition}
  A morphism is called \textit{strongly non-signalling} if it factorises as in equation~\eqref{eq:factor-bell} for some causal morphisms $\Psi_A$, $\Psi_B$, and $\rho$.
\end{definition}

It was shown in \cite{BGNP} that there indeed exist quantum channels which are non-signalling but not strongly non-signalling (conditions referred to as `causal' and `localizable' in \cite{BGNP}, respectively).

\section{Second-order causality}

Recently, frameworks have been proposed to discuss quantum correlations
 which do not necessarily have a fixed causal ordering~\cite{ViennaIndef,QSwitch}.
Both of these frameworks rely on the notion of a `higher-order quantum channel'~\cite{QSupermaps}, i.e. a mapping which sends channels to channels. In this section, we will provide a characterisation of such a map in any compact closed category with discarding.

As is the usual trick in a compact-closed category, we can obtain higher-order maps by first turning first order maps into states by `bending up' the input wire:
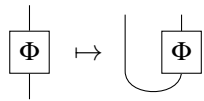
\[ \boxmap{\Phi} \ \ \mapsto\ \  \begin{tikzpicture}
	\begin{pgfonlayer}{nodelayer}
		\node [style=small box] (0) at (0.75, 0) {$\Phi$};
		\node [style=none] (1) at (0.75, -0.5) {};
		\node [style=none] (2) at (-0.75, -0.5) {};
		\node [style=none] (3) at (-0.75, 1) {};
		\node [style=none] (4) at (0.75, 1) {};
	\end{pgfonlayer}
	\begin{pgfonlayer}{edgelayer}
		\draw (4.center) to (1.center);
		\draw [in=-90, out=-90, looseness=1.25] (1.center) to (2.center);
		\draw (2.center) to (3.center);
	\end{pgfonlayer}
\end{tikzpicture} \]
This bending is sometimes called \textit{process-state duality}, which induces a bijection between:
\begin{equation}\label{eq:p-s}
  \left\{\ \textrm{processes } \vphantom{\tilde\Phi}\Phi : A \to B\ \right\}
\ \cong\ 
\left\{\ \textrm{states } \tilde\Phi : I \to A \otimes B\ \right\}
\end{equation}

Hence, we can express a map which sends a process of type $A_1 \to A_2$ to a process of type $B_1 \to B_2$ as a map of the form:
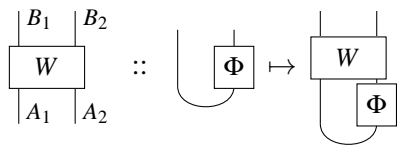
\[
\begin{tikzpicture}
	\begin{pgfonlayer}{nodelayer}
		\node [style=none] (0) at (-0.75, -0.25) {};
		\node [style=none] (1) at (0.75, 0.25) {};
		\node [style=none] (2) at (0.75, -1.5) {};
		\node [style=none] (3) at (0.75, -0.25) {};
		\node [style=none] (4) at (0.75, 1.5) {};
		\node [style=medium box] (5) at (0, 0) {$W$};
		\node [style=none] (6) at (-0.75, 1.5) {};
		\node [style=none] (7) at (-0.75, 0.25) {};
		\node [style=none] (8) at (-0.75, -1.5) {};
		\node [style=right label] (9) at (-0.5, -1.25) {$A_1$};
		\node [style=right label] (10) at (1, -1.25) {$A_2$};
		\node [style=right label] (11) at (1, 1.25) {$B_2$};
		\node [style=right label] (12) at (-0.5, 1.25) {$B_1$};
	\end{pgfonlayer}
	\begin{pgfonlayer}{edgelayer}
		\draw (8.center) to (0.center);
		\draw (7.center) to (6.center);
		\draw (2.center) to (3.center);
		\draw (1.center) to (4.center);
	\end{pgfonlayer}
\end{tikzpicture}
\ ::\ \ \ 
\begin{tikzpicture}
	\begin{pgfonlayer}{nodelayer}
		\node [style=small box] (0) at (0.75, 0) {$\Phi$};
		\node [style=none] (1) at (0.75, -0.5) {};
		\node [style=none] (2) at (-0.75, -0.5) {};
		\node [style=none] (3) at (-0.75, 1) {};
		\node [style=none] (4) at (0.75, 1) {};
	\end{pgfonlayer}
	\begin{pgfonlayer}{edgelayer}
		\draw (4.center) to (1.center);
		\draw [in=-90, out=-90, looseness=1.25] (1.center) to (2.center);
		\draw (2.center) to (3.center);
	\end{pgfonlayer}
\end{tikzpicture}\ \mapsto\ \begin{tikzpicture}
	\begin{pgfonlayer}{nodelayer}
		\node [style=none] (0) at (0.75, 0.5) {};
		\node [style=none] (1) at (0.75, 1.5) {};
		\node [style=medium box] (2) at (0, 0.25) {$W$};
		\node [style=none] (3) at (-0.75, 1.5) {};
		\node [style=none] (4) at (-0.75, 0.5) {};
		\node [style=none] (5) at (-0.75, 0) {};
		\node [style=none] (6) at (0.75, -1.5) {};
		\node [style=none] (7) at (-0.75, -1.5) {};
		\node [style=small box] (8) at (0.75, -1) {$\Phi$};
		\node [style=none] (9) at (0.75, 0) {};
	\end{pgfonlayer}
	\begin{pgfonlayer}{edgelayer}
		\draw (4.center) to (3.center);
		\draw (0.center) to (1.center);
		\draw (9.center) to (6.center);
		\draw [in=-90, out=-90, looseness=1.25] (6.center) to (7.center);
		\draw (7.center) to (5.center);
	\end{pgfonlayer}
\end{tikzpicture}
\]

\begin{definition}
  A process is called \textit{second-order causal} (SOC) if it sends causal processes (encoded via process-state duality \eqref{eq:p-s}) to causal processes. Diagrammatically, $W$ is SOC if for all $\Phi$:
  \begin{equation}\label{eq:soc}
    \sboxdiscard{\Phi}\ =\ \sdiscard
    \ \ \implies\ \ 
    \begin{tikzpicture}
	\begin{pgfonlayer}{nodelayer}
		\node [style=none] (0) at (0.75, 0.5) {};
		\node [style=upground] (1) at (0.75, 1.5) {};
		\node [style=medium box] (2) at (0, 0.25) {$W$};
		\node [style=none] (3) at (-0.75, 0) {};
		\node [style=none] (4) at (0.75, -1.5) {};
		\node [style=none] (5) at (-0.75, -1.5) {};
		\node [style=small box] (6) at (0.75, -1) {$\Phi$};
		\node [style=none] (7) at (0.75, 0) {};
		\node [style=none] (8) at (-2.25, -2) {};
		\node [style=none] (9) at (-2.25, 0.75) {};
		\node [style=none] (10) at (-0.75, 0.5) {};
		\node [style=none] (11) at (-0.75, 0.75) {};
		\node [style=none] (12) at (2.25, 0) {$=$};
		\node [style=none] (13) at (3.75, -0.75) {};
		\node [style=upground] (14) at (3.75, 0.5) {};
	\end{pgfonlayer}
	\begin{pgfonlayer}{edgelayer}
		\draw (0.center) to (1);
		\draw (7.center) to (4.center);
		\draw [in=-90, out=-90, looseness=1.25] (4.center) to (5.center);
		\draw (5.center) to (3.center);
		\draw (10.center) to (11.center);
		\draw [in=90, out=90, looseness=1.25] (11.center) to (9.center);
		\draw (9.center) to (8.center);
		\draw (14) to (13.center);
	\end{pgfonlayer}
\end{tikzpicture}
  \end{equation}
\end{definition}

It is often more enlightening to write SOC maps using `comb' notation (cf. \cite{PaviaIndef}):
\begin{equation}\label{eq:soc-comb}
  \begin{tikzpicture}
	\begin{pgfonlayer}{nodelayer}
		\node [style=none] (0) at (-0.75, -0.25) {};
		\node [style=none] (1) at (0.75, 0.25) {};
		\node [style=none] (2) at (0.75, -1.5) {};
		\node [style=none] (3) at (0.75, -0.25) {};
		\node [style=none] (4) at (0.75, 1.5) {};
		\node [style=medium box] (5) at (0, 0) {$W$};
		\node [style=none] (6) at (-0.75, 1.5) {};
		\node [style=none] (7) at (-0.75, 0.25) {};
		\node [style=none] (8) at (-0.75, -1.5) {};
		\node [style=right label] (9) at (-0.5, -1.25) {$A_1$};
		\node [style=right label] (10) at (1, -1.25) {$A_2$};
		\node [style=right label] (11) at (1, 1.25) {$B_2$};
		\node [style=right label] (12) at (-0.5, 1.25) {$B_1$};
	\end{pgfonlayer}
	\begin{pgfonlayer}{edgelayer}
		\draw (8.center) to (0.center);
		\draw (7.center) to (6.center);
		\draw (2.center) to (3.center);
		\draw (1.center) to (4.center);
	\end{pgfonlayer}
\end{tikzpicture}
\qquad\leadsto\qquad\quad
\begin{tikzpicture}
	\begin{pgfonlayer}{nodelayer}
		\node [style=none] (0) at (1.25, 2.25) {};
		\node [style=none] (1) at (1.25, 1.25) {};
		\node [style=none] (2) at (-0.25, 1.25) {};
		\node [style=none] (3) at (-0.25, -1.25) {};
		\node [style=none] (4) at (1.25, -1.25) {};
		\node [style=none] (5) at (1.25, -2.25) {};
		\node [style=none] (6) at (-1.75, -2.25) {};
		\node [style=none] (7) at (-1.75, 2.25) {};
		\node [style=none] (8) at (0.5, 1.25) {};
		\node [style=none] (9) at (0.5, 0.5) {};
		\node [style=none] (10) at (0.5, -0.5) {};
		\node [style=none] (11) at (0.5, -1.25) {};
		\node [style=right label] (12) at (0.75, 0.5) {$A_2$};
		\node [style=right label] (13) at (0.75, -0.75) {$A_1$};
		\node [style=none] (14) at (0, -3) {};
		\node [style=right label] (15) at (0.25, -3) {$B_1$};
		\node [style=none] (16) at (0, -2.25) {};
		\node [style=none] (17) at (0, 3) {};
		\node [style=right label] (18) at (0.25, 2.75) {$B_2$};
		\node [style=none] (19) at (0, 2.25) {};
		\node [style=none] (20) at (-1, 0) {$W$};
	\end{pgfonlayer}
	\begin{pgfonlayer}{edgelayer}
		\draw (0.center) to (1.center);
		\draw (1.center) to (2.center);
		\draw (2.center) to (3.center);
		\draw (3.center) to (4.center);
		\draw (4.center) to (5.center);
		\draw (5.center) to (6.center);
		\draw (6.center) to (7.center);
		\draw (7.center) to (0.center);
		\draw (8.center) to (9.center);
		\draw (10.center) to (11.center);
		\draw (16.center) to (14.center);
		\draw (19.center) to (17.center);
	\end{pgfonlayer}
\end{tikzpicture}
\end{equation}
Then processes are composed from the inside-out, rather than bottom-to-top. Hence, \eqref{eq:soc} becomes:
\[ \sboxdiscard{\Phi}\ =\ \sdiscard
    \qquad\implies\qquad
    \begin{tikzpicture}
	\begin{pgfonlayer}{nodelayer}
		\node [style=none] (0) at (1.25, 1.75) {};
		\node [style=none] (1) at (1.25, 0.75) {};
		\node [style=none] (2) at (-0.25, 0.75) {};
		\node [style=none] (3) at (-0.25, -0.75) {};
		\node [style=none] (4) at (1.25, -0.75) {};
		\node [style=none] (5) at (1.25, -1.75) {};
		\node [style=none] (6) at (-1.75, -1.75) {};
		\node [style=none] (7) at (-1.75, 1.75) {};
		\node [style=none] (8) at (0.5, 0.75) {};
		\node [style=none] (9) at (0.5, 0.5) {};
		\node [style=none] (10) at (0.5, -0.5) {};
		\node [style=none] (11) at (0.5, -0.75) {};
		\node [style=none] (12) at (0, -2.5) {};
		\node [style=none] (13) at (0, -1.75) {};
		\node [style=upground] (14) at (0, 2.5) {};
		\node [style=none] (15) at (0, 1.75) {};
		\node [style=none] (16) at (-1, 0) {$W$};
		\node [style=none] (17) at (2.5, 0) {$=$};
		\node [style=small box] (18) at (0.5, 0) {$\Phi$};
		\node [style=none] (19) at (4.25, -0.75) {};
		\node [style=upground] (20) at (4.25, 0.5) {};
	\end{pgfonlayer}
	\begin{pgfonlayer}{edgelayer}
		\draw (0.center) to (1.center);
		\draw (1.center) to (2.center);
		\draw (2.center) to (3.center);
		\draw (3.center) to (4.center);
		\draw (4.center) to (5.center);
		\draw (5.center) to (6.center);
		\draw (6.center) to (7.center);
		\draw (7.center) to (0.center);
		\draw (8.center) to (9.center);
		\draw (10.center) to (11.center);
		\draw (13.center) to (12.center);
		\draw (15.center) to (14);
		\draw (19.center) to (20);
	\end{pgfonlayer}
\end{tikzpicture} \]

However, processes with just one `hole' are not that interesting, so we will consider a more interesting kind of second-order causal map, which has two holes:

\begin{definition}
  A process:
  \[ W : (A_1 \otimes A_2) \otimes (B_1 \otimes B_2) \to C_1 \otimes C_2 \]
  is called \textit{bipartite second-order causal} ($\textrm{SOC}_2$) if for all causal $\Phi_A, \Phi_B$:
  \ctikzfig{SOC2}
  is causal.
\end{definition}

Particularly simple examples of \SOCt maps simply wire $\Phi_A$ together with $\Phi_B$ in some order:
\begin{equation}\label{eq:def-cs}
  \begin{tikzpicture}
	\begin{pgfonlayer}{nodelayer}
		\node [style=none] (0) at (-0.75, 1.75) {};
		\node [style=none] (1) at (-0.75, 0.75) {};
		\node [style=none] (2) at (-2.25, 0.75) {};
		\node [style=none] (3) at (-2.25, -0.75) {};
		\node [style=none] (4) at (-0.75, -0.75) {};
		\node [style=none] (5) at (-0.75, -1.75) {};
		\node [style=none] (6) at (-5.25, -1.75) {};
		\node [style=none] (7) at (-5.25, 1.75) {};
		\node [style=none] (8) at (-1.5, 0.75) {};
		\node [style=none] (9) at (-1.5, 0.5) {};
		\node [style=none] (10) at (-1.5, -0.5) {};
		\node [style=none] (11) at (-1.5, -0.75) {};
		\node [style=none] (12) at (-3, -2.5) {};
		\node [style=none] (13) at (-3, -1.75) {};
		\node [style=none] (14) at (-3, 2.5) {};
		\node [style=none] (15) at (-3, 1.75) {};
		\node [style=none] (16) at (-3.75, 0.75) {};
		\node [style=none] (17) at (-4.5, 0.75) {};
		\node [style=none] (18) at (-4.5, -0.5) {};
		\node [style=none] (19) at (-4.5, 0.5) {};
		\node [style=none] (20) at (-4.5, -0.75) {};
		\node [style=none] (21) at (-5.25, -0.75) {};
		\node [style=none] (22) at (-5.25, 0.75) {};
		\node [style=none] (23) at (-3.75, -0.75) {};
		\node [style=none] (24) at (-4.5, 1) {};
		\node [style=none] (25) at (-3.25, 1) {};
		\node [style=none] (26) at (-2.75, -1) {};
		\node [style=none] (27) at (-1.5, -1) {};
	\end{pgfonlayer}
	\begin{pgfonlayer}{edgelayer}
		\draw [style=dashed edge] (0.center) to (1.center);
		\draw [style=dashed edge] (1.center) to (2.center);
		\draw [style=dashed edge] (2.center) to (3.center);
		\draw [style=dashed edge] (3.center) to (4.center);
		\draw [style=dashed edge] (4.center) to (5.center);
		\draw [style=dashed edge] (5.center) to (6.center);
		\draw [style=dashed edge] (7.center) to (0.center);
		\draw (8.center) to (9.center);
		\draw (10.center) to (11.center);
		\draw (13.center) to (12.center);
		\draw (15.center) to (14.center);
		\draw [style=dashed edge] (22.center) to (16.center);
		\draw [style=dashed edge] (16.center) to (23.center);
		\draw [style=dashed edge] (23.center) to (21.center);
		\draw (17.center) to (19.center);
		\draw (18.center) to (20.center);
		\draw [style=dashed edge] (6.center) to (21.center);
		\draw [style=dashed edge] (22.center) to (7.center);
		\draw [in=-90, out=90, looseness=1.00] (13.center) to (20.center);
		\draw [in=-90, out=90, looseness=1.00] (8.center) to (15.center);
		\draw (17.center) to (24.center);
		\draw [in=105, out=90, looseness=1.00] (24.center) to (25.center);
		\draw [in=105, out=-75, looseness=1.00] (25.center) to (26.center);
		\draw [in=-90, out=-75, looseness=1.00] (26.center) to (27.center);
		\draw (27.center) to (11.center);
	\end{pgfonlayer}
\end{tikzpicture} \qquad\qquad \begin{tikzpicture}
	\begin{pgfonlayer}{nodelayer}
		\node [style=none] (0) at (-5.25, 1.75) {};
		\node [style=none] (1) at (-5.25, 0.75) {};
		\node [style=none] (2) at (-3.75, 0.75) {};
		\node [style=none] (3) at (-3.75, -0.75) {};
		\node [style=none] (4) at (-5.25, -0.75) {};
		\node [style=none] (5) at (-5.25, -1.75) {};
		\node [style=none] (6) at (-0.75, -1.75) {};
		\node [style=none] (7) at (-0.75, 1.75) {};
		\node [style=none] (8) at (-4.5, 0.75) {};
		\node [style=none] (9) at (-4.5, 0.5) {};
		\node [style=none] (10) at (-4.5, -0.5) {};
		\node [style=none] (11) at (-4.5, -0.75) {};
		\node [style=none] (12) at (-3, -2.5) {};
		\node [style=none] (13) at (-3, -1.75) {};
		\node [style=none] (14) at (-3, 2.5) {};
		\node [style=none] (15) at (-3, 1.75) {};
		\node [style=none] (16) at (-2.25, 0.75) {};
		\node [style=none] (17) at (-1.5, 0.75) {};
		\node [style=none] (18) at (-1.5, -0.5) {};
		\node [style=none] (19) at (-1.5, 0.5) {};
		\node [style=none] (20) at (-1.5, -0.75) {};
		\node [style=none] (21) at (-0.75, -0.75) {};
		\node [style=none] (22) at (-0.75, 0.75) {};
		\node [style=none] (23) at (-2.25, -0.75) {};
		\node [style=none] (24) at (-1.5, 1) {};
		\node [style=none] (25) at (-2.75, 1) {};
		\node [style=none] (26) at (-3.25, -1) {};
		\node [style=none] (27) at (-4.5, -1) {};
	\end{pgfonlayer}
	\begin{pgfonlayer}{edgelayer}
		\draw [style=dashed edge] (0.center) to (1.center);
		\draw [style=dashed edge] (1.center) to (2.center);
		\draw [style=dashed edge] (2.center) to (3.center);
		\draw [style=dashed edge] (3.center) to (4.center);
		\draw [style=dashed edge] (4.center) to (5.center);
		\draw [style=dashed edge] (5.center) to (6.center);
		\draw [style=dashed edge] (7.center) to (0.center);
		\draw (8.center) to (9.center);
		\draw (10.center) to (11.center);
		\draw (13.center) to (12.center);
		\draw (15.center) to (14.center);
		\draw [style=dashed edge] (22.center) to (16.center);
		\draw [style=dashed edge] (16.center) to (23.center);
		\draw [style=dashed edge] (23.center) to (21.center);
		\draw (17.center) to (19.center);
		\draw (18.center) to (20.center);
		\draw [style=dashed edge] (6.center) to (21.center);
		\draw [style=dashed edge] (22.center) to (7.center);
		\draw [in=-90, out=90, looseness=1.00] (13.center) to (20.center);
		\draw [in=-90, out=90, looseness=1.00] (8.center) to (15.center);
		\draw (17.center) to (24.center);
		\draw [in=75, out=90, looseness=1.00] (24.center) to (25.center);
		\draw [in=75, out=-105, looseness=1.00] (25.center) to (26.center);
		\draw [in=-90, out=-105, looseness=1.00] (26.center) to (27.center);
		\draw (27.center) to (11.center);
	\end{pgfonlayer}
\end{tikzpicture}
\end{equation}
However, interestingly, the order that they are wired together is hidden if we treat $W$ as a black box,
and it can be shown (see for example~\cite{ViennaIndef}) that we can even define \SOCt maps which don't admit any fixed causal order.

It is natural to ask whether separate notions of \textit{bipartite} (or more generally, $n$-partite) second-order causal maps is really necessary. We could, after all, define an SOC map as in \eqref{eq:soc-comb} where $A_1 := X_1 \otimes Y_1$ and $A_2 := X_2 \otimes Y_2$, i.e. of the form:
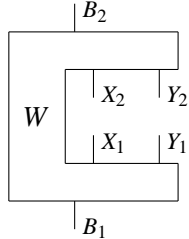
\ctikzfig{SOC-try2}
but this is very restrictive since it needs to send \underline{any} causal map to a causal map, rather than just separable ones. In fact, the simple example of an \SOCt map which wires two processes together in some fixed order is already not SOC. Suppose for instance that we plug a (non-separable) swap map into the leftmost process in~\eqref{eq:def-cs}:
\ctikzfig{SOC2-wire-swap}
Then we introduce a loop, which for most categories $\mathcal C$ (including CP-maps) will immediately kill normalisation, and hence causality.

\begin{definition}\label{def:enough-states}
  We say a category $\mathcal C$ has \textit{enough causal states} if:
  \[
  \left( \ \forall \rho \ \textrm{causal}\ \ .\ \ \kpointmap[style1=point,style2=small box]{\rho}{\Phi} = \kpointmap[style1=point,style2=small box]{\rho}{\Phi'}\  \right)
  \implies 
  \Phi = \Phi'
  \]
\end{definition}

Since $\mathcal C$ is compact closed, we can prove that if $\mathcal C$ has enough states, it also has enough \textit{separable} causal states:
\[
\left( \ \forall \rho_1, \ldots, \rho_n \ \textrm{causal}\ \ .\ \ \begin{tikzpicture}
	\begin{pgfonlayer}{nodelayer}
		\node [style=medium box] (0) at (-2.25, 0.5) {$\Phi$};
		\node [style=point] (1) at (-3, -1.25) {$\rho_1$};
		\node [style=point] (2) at (-1.5, -1.25) {$\rho_n$};
		\node [style=none] (3) at (-1.5, 1.5) {};
		\node [style=none] (4) at (-3, 1.5) {};
		\node [style=none] (5) at (0, 0) {$=$};
		\node [style=none] (6) at (1.5, 1.5) {};
		\node [style=point] (7) at (1.5, -1.25) {$\rho_1$};
		\node [style=medium box] (8) at (2.25, 0.5) {$\Phi'$};
		\node [style=none] (9) at (3, 1.5) {};
		\node [style=point] (10) at (3, -1.25) {$\rho_n$};
		\node [style=none] (11) at (-2.25, -0.5) {$\cdots$};
		\node [style=none] (12) at (2.25, -0.5) {$\cdots$};
		\node [style=none] (13) at (2.25, 1.5) {$\cdots$};
		\node [style=none] (14) at (-2.25, 1.5) {$\cdots$};
	\end{pgfonlayer}
	\begin{pgfonlayer}{edgelayer}
		\draw (1) to (4.center);
		\draw (2) to (3.center);
		\draw (7) to (6.center);
		\draw (10) to (9.center);
	\end{pgfonlayer}
\end{tikzpicture}\  \right)
  \implies 
  \Phi = \Phi'
\]
because we can simply apply Definition~\ref{def:enough-states} one input at a time, via:
\[
\begin{tikzpicture}
	\begin{pgfonlayer}{nodelayer}
		\node [style=medium box] (0) at (-1, 0.25) {$\Phi$};
		\node [style=point] (1) at (-1.75, -1) {$\rho$};
		\node [style=none] (2) at (-0.25, -1) {};
		\node [style=none] (3) at (-0.25, 1.25) {};
		\node [style=none] (4) at (-1.75, 1.25) {};
		\node [style=none] (5) at (1.25, 0) {$=$};
		\node [style=none] (6) at (2.75, 1.25) {};
		\node [style=none] (7) at (4.25, -1) {};
		\node [style=point] (8) at (2.75, -1) {$\rho$};
		\node [style=medium box] (9) at (3.5, 0.25) {$\Phi'$};
		\node [style=none] (10) at (4.25, 1.25) {};
	\end{pgfonlayer}
	\begin{pgfonlayer}{edgelayer}
		\draw (1) to (4.center);
		\draw (2.center) to (3.center);
		\draw (8) to (6.center);
		\draw (7.center) to (10.center);
	\end{pgfonlayer}
\end{tikzpicture} \qquad \iff \qquad
\begin{tikzpicture}
	\begin{pgfonlayer}{nodelayer}
		\node [style=medium box] (0) at (-1.75, 0.25) {$\Phi$};
		\node [style=point] (1) at (-2.5, -1) {$\rho$};
		\node [style=none] (2) at (-1, -1) {};
		\node [style=none] (3) at (-1, 1.25) {};
		\node [style=none] (4) at (-2.5, 1.25) {};
		\node [style=none] (5) at (1.25, 0) {$=$};
		\node [style=none] (6) at (2.75, 1.25) {};
		\node [style=point] (7) at (2.75, -1) {$\rho$};
		\node [style=medium box] (8) at (3.5, 0.25) {$\Phi'$};
		\node [style=none] (9) at (0, 1.25) {};
		\node [style=none] (10) at (0, -1) {};
		\node [style=none] (11) at (5.25, -1) {};
		\node [style=none] (12) at (4.25, 1.25) {};
		\node [style=none] (13) at (4.25, -1) {};
		\node [style=none] (14) at (5.25, 1.25) {};
	\end{pgfonlayer}
	\begin{pgfonlayer}{edgelayer}
		\draw (1) to (4.center);
		\draw (2.center) to (3.center);
		\draw (7) to (6.center);
		\draw (10.center) to (9.center);
		\draw [in=-90, out=-90, looseness=1.25] (2.center) to (10.center);
		\draw (13.center) to (12.center);
		\draw (11.center) to (14.center);
		\draw [in=-90, out=-90, looseness=1.25] (13.center) to (11.center);
	\end{pgfonlayer}
\end{tikzpicture}
\]
%For the remainder of this section, we will assume $\mathcal C$ has enough causal states.

\section{Second-order causality and non-signalling}

Now we are ready to prove the main theorem of this extended abstract and give a simple corollary.

\begin{theorem}\label{thm:complete-SOC2}
  If a process in a category with enough causal states is \SOCt, then it is `completely' \SOCt in the sense that, for any causal processes:
  \[
  \Phi_A : A'_1 \otimes A_1 \to A'_2 \otimes A_2 
  \qquad\qquad
  \Phi_B : B'_1 \otimes B_1 \to B'_2 \otimes B_2 
  \]
  the process:
  \ctikzfig{compl-SOC2}
  is causal.
\end{theorem}

\begin{proof}
  For any causal states $\rho_A, \rho_B$, the following processes are causal:
  \begin{equation}\label{eq:causal-labs}
    \begin{tikzpicture}
	\begin{pgfonlayer}{nodelayer}
		\node [style=medium box] (0) at (-2.25, 0.25) {$\Phi_A$};
		\node [style=point] (1) at (-3, -1) {$\rho_A$};
		\node [style=none] (2) at (-1.5, -1) {};
		\node [style=none] (3) at (-1.5, 1.25) {};
		\node [style=upground] (4) at (-3, 1.5) {};
		\node [style=none] (5) at (3, 1.25) {};
		\node [style=none] (6) at (3, -1) {};
		\node [style=medium box] (7) at (3.75, 0.25) {$\Phi_B$};
		\node [style=upground] (8) at (4.5, 1.5) {};
		\node [style=point] (9) at (4.5, -1) {$\rho_B$};
	\end{pgfonlayer}
	\begin{pgfonlayer}{edgelayer}
		\draw (1) to (4);
		\draw (2.center) to (3.center);
		\draw (6.center) to (5.center);
		\draw (9) to (8);
	\end{pgfonlayer}
\end{tikzpicture}
  \end{equation}
  which can be seen just by discarding the respective outputs and applying causality of $\Phi_A, \Phi_B, \rho_A$ and $\rho_B$ individually. Then, if $W$ is \SOCt, plugging in the causal maps~\eqref{eq:causal-labs} yields a causal map. Hence, for any $\rho_W$, we have:
  \ctikzfig{SOC2-pf1}
  Since the process above agrees with discarding for all $\rho_A, \rho_B, \rho_W$:
  \[ \begin{tikzpicture}
	\begin{pgfonlayer}{nodelayer}
		\node [style=upground] (0) at (-1.75, 0.75) {};
		\node [style=point] (1) at (-1.75, -0.75) {$\rho_A$};
		\node [style=point] (2) at (0, -0.75) {$\rho_W\!$};
		\node [style=upground] (3) at (0, 0.75) {};
		\node [style=point] (4) at (1.75, -0.75) {$\rho_B$};
		\node [style=upground] (5) at (1.75, 0.75) {};
	\end{pgfonlayer}
	\begin{pgfonlayer}{edgelayer}
		\draw (1) to (0);
		\draw (2) to (3);
		\draw (4) to (5);
	\end{pgfonlayer}
\end{tikzpicture}\ =\ \emptydiag \]
  we can conclude, using the fact that $\mathcal C$ has enough causal states (and hence enough separable causal states) that:
  \ctikzfig{compl-SOC2-pf}
\end{proof}

We can now show that \SOCt processes not only preserve causality for separable processes, but also for strongly non-signalling processes:

\begin{corollary}\label{cor:strongsign}
  If a process $W$
  is \SOCt, then it sends every causal, strongly non-signalling process:
  \[ \Phi : A_1 \otimes B_1 \to A_2 \otimes B_2 \]
  to a causal process.
\end{corollary}

\begin{proof}
  If $\Phi$ is strongly non-signalling, it factors as in \eqref{eq:factor-bell}. Then:
  \ctikzfig{W-nonsig}
\end{proof}

In~\cite{QSwitch} it is shown that preserving causality for product channels is equivalent to preserving causality for all non-signalling channels. This can be shown straight-forwardly in the concrete case of CP-maps using the fact that non-signalling channels always arise as affine combinations of separable channels. One could therefore extend the proof above to work for all non-signalling processes if we replace $\rho$ with a `pseudo-state' given by, e.g.
\[ \begin{tikzpicture}
	\begin{pgfonlayer}{nodelayer}
		\node [style=none] (0) at (-0.75, -0.25) {};
		\node [style=none] (1) at (0.75, -0.25) {};
		\node [style=kpoint sc, minimum width=5 mm, fill={gray!50!white}] (2) at (0, -0.5) {$\bm r$};
		\node [style=none] (3) at (0.75, 0.5) {};
		\node [style=none] (4) at (-0.75, 0.5) {};
	\end{pgfonlayer}
	\begin{pgfonlayer}{edgelayer}
		\draw (1.center) to (3.center);
		\draw (0.center) to (4.center);
	\end{pgfonlayer}
\end{tikzpicture} \  :=\  \sum_i r_i\ \ketbra{i}{i} \otimes \ketbra{i}{i} \]
for (possibly negative) coefficients $r_i$ summing to $1$. Then we still have:
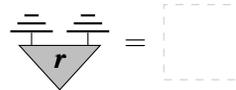
\[ \begin{tikzpicture}
	\begin{pgfonlayer}{nodelayer}
		\node [style=none] (0) at (-0.75, -0.25) {};
		\node [style=none] (1) at (0.75, -0.25) {};
		\node [style=kpoint sc, minimum width=5 mm, fill={gray!50!white}] (2) at (0, -0.5) {$\bm r$};
		\node [style=upground] (3) at (0.75, 0.5) {};
		\node [style=upground] (4) at (-0.75, 0.5) {};
	\end{pgfonlayer}
	\begin{pgfonlayer}{edgelayer}
		\draw (1.center) to (3);
		\draw (0.center) to (4);
	\end{pgfonlayer}
\end{tikzpicture}\ =\ \emptydiag \]
and we can furthermore realise any affine combination of separable CP-maps (hence any non-signalling channel) via:
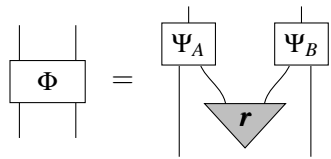
\ctikzfig{factor-bell-pseudo}
Then the proof of Corollary \ref{cor:strongsign} proceeds identically, replacing $\rho$ with $\bm r$. However, this has the undesirable property that we have to go outside of the category of `physically realisable' processes to get this (non-positive) pseudo-state $\bm r$. Whether one can give a fully diagrammatic proof without resorting to such tricks is an open question.

\section{Acknowledgements}
We would like to thank everyone from the 2016 Joint McGill-Oxford Workshop on Causality in Quantum Foundations for fruitful discussions as well as the anonymous reviewers for their useful comments and references. The research leading to these results has received funding from the European Research Council under the European Union's Seventh Framework Programme (FP7/2007-2013) / ERC grant agreement no. 320571

%  have a family of causal morphisms $\Phi_i$, then the choice of $i$ could in principle affect the outputs of the causal morphisms shown in white below:
% \ctikzfig{light-cone}
% ...but not those in grey. For example, focusing only on the output of morphism $\Psi$ gives:
% \ctikzfig{light-cone2}
% However, if we allow non-causal morphism (e.g. the `cap'), this is no longer the case:

\bibliographystyle{eptcs}
\bibliography{main}

\end{document}